\documentclass[msom,nonblindrev]{informs4}

\DoubleSpacedXI

\usepackage{amsmath,amssymb,amsfonts}
\usepackage{threeparttable}
\usepackage{natbib}
 \bibpunct[, ]{(}{)}{,}{a}{}{,}%
 \def\bibfont{\normalsize}%

\usepackage{anyfontsize}
\usepackage{rotating}
\usepackage{fancyvrb}
\usepackage{endnotes}
\usepackage{rotating}
\usepackage{cleveref}
\crefname{lemma}{Lemma}{Lemmas}  
\crefname{proposition}{Proposition}{Propositions}

\usepackage{multirow}
\usepackage{tabu, booktabs}
\usepackage{makecell}
\usepackage{subfigure}
\usepackage[table,xcdraw]{xcolor}
\usepackage{bbm}

 \usepackage{amssymb}
\usepackage{tabularx}
\usepackage{booktabs}
\usepackage{mathtools}
\usepackage{float}

\usepackage{pifont} 

\TheoremsNumberedThrough     
\ECRepeatTheorems

\EquationsNumberedThrough    

\newcommand{\edit}[1]{{\color{black} #1}}

\usepackage{array}
\usepackage{booktabs}
\usepackage{makecell}
\usepackage{tabularx}

\usepackage{xcolor}

\makeatletter
\newcommand{\bluebox}[1]{%
  \begingroup
    \color{black}%
    \everymath{\color{black}}%

    \let\BB@oldcaption\caption
    \renewcommand{\caption}[2][]{%
      \BB@oldcaption[\textcolor{black}{##1}]{\textcolor{black}{##2}}%
    }%

    \let\BB@oldfnum@table\fnum@table
    \let\BB@oldfnum@figure\fnum@figure
    \renewcommand{\fnum@table}{\textcolor{black}{\BB@oldfnum@table}}%
    \renewcommand{\fnum@figure}{\textcolor{black}{\BB@oldfnum@figure}}%

    #1
  \endgroup
}
\makeatother

\begin{document}

\RUNAUTHOR{Zhao et al.} 

\RUNTITLE{Improving Human Supervision of Algorithms}

\TITLE{A Simple Solution to Improving Human Supervision of Algorithms: Evidence from Smart Vending}

\ARTICLEAUTHORS{%
    \AUTHOR{Minda Zhao}
    \AFF{H. Milton Stewart School of Industrial and Systems Engineering, Georgia Institute of Technology, \EMAIL{mindazhao@gatech.edu}}

    \AUTHOR{Brian Rongqing Han}
    \AFF{Gies College of Business, University of Illinois Urbana-Champaign, \EMAIL{brianhan@illinois.edu}}

    \AUTHOR{Xin Chen}
    \AFF{Naveen Jindal School of Management, University of Texas at Dallas, \EMAIL{xin.chen@utdallas.edu}}

    \AUTHOR{Tao Zhu}
    \AFF{Shenzhen Fengyi Technology Co., Ltd., \EMAIL{zhutao@mail.feng1.com}}
}

\ABSTRACT{
\textbf{\textit{Problem definition:}}
Organizations increasingly deploy autonomous artificial intelligence (AI) systems for operational decisions, such as inventory replenishment. Yet granting full override discretion can degrade performance because of human bias and judgment errors, while prohibiting overrides may overlook valuable private information. This raises a key question: How should override discretion be structured to improve human supervision of autonomous AI? 
\textbf{\textit{Methodology/results:}}
We propose a constrained override policy that limits the number of overrides within each decision episode. When the constraint binds, humans prioritize overrides they consider most valuable, a process we call \textit{selective filtering}. We evaluate the policy in a smart-vending field experiment, randomly assigning $553$ workers to no, free, or constrained downward override, with the constrained policy capped at two stock-keeping units (SKUs) per machine replenishment. Relative to the no-override policy, free override reduces inventory by $1.95$ log points but also lowers sales by $1.19$ log points, whereas constrained override reduces inventory by $1.28$ log points without a statistically significant sales decline. Local average treatment effect estimates support this contrast. Constrained override also reduces inventory without statistically significant sales losses among experienced workers, for top-sales and growth-stage SKUs, and with greater machine--SKU familiarity. Simulations suggest calibrating the cap to the firm's relative margin, worker experience, and machine--SKU familiarity.
\textbf{\textit{Managerial implications:}}
Constrained override offers a simple and scalable way to use workers' private information in inventory replenishment. Firms can limit override discretion through the replenishment interface without redesigning the algorithm, identifying useful overrides in advance, providing individualized feedback, or retraining workers.
}


\maketitle

\section{Introduction}
\label{section: introduction}
Data-driven decision-making is transforming businesses by automating complex and recurring tasks previously reliant on human judgment \citep{brynjolfsson2016rapid, bar2024helping}. Algorithms and artificial intelligence (AI) systems\endnote{Throughout this paper, we use the terms \emph{algorithm} and \emph{AI} interchangeably.} are now routinely employed in critical decision-making domains such as drug discovery \citep{lou2021ai} and digital content recommendation \citep{chen2024impact}. Despite substantial advances in algorithmic capabilities, many organizations still grant frontline employees discretionary power to override algorithmic recommendations. Such override rights reflect a widespread managerial belief that human workers possess contextual or local information unavailable to algorithms, which can potentially enhance decisions when algorithms falter.

This belief is especially plausible in operational settings, where frontline workers remain close to physical processes and local demand. Workers' own forecasts can shape the use of retail-assortment algorithms \citep{kawaguchi2021will}; package-specific knowledge can explain some deviations from warehouse bin-packing prescriptions \citep{sun2022predicting}; and managers use contextual information when adjusting statistical forecasts \citep{fildes2009effective}. In automated replenishment, store managers also incorporate in-store handling costs and sales opportunities omitted by the system \citep{van2010ordering}. Inventory management brings these sources of information together in a particularly demanding environment: firms must forecast local demand, balance excess inventory against stockouts and lost sales, and plan across many SKUs. Unlike a forecast adjustment alone, a replenishment override immediately changes physical inventory and thus both inventory costs and the ability to satisfy demand. Workers may identify conditions the algorithm misses, but must decide where among many recommendations to intervene, making human supervision potentially valuable but difficult to govern.

Evidence on the performance of human overrides is mixed. While a substantial body of literature (e.g., \citealt{ibanez2018discretionary,kesavan2020field, liang2024examining}) suggests that human overrides can impair the performance of autonomous algorithmic decisions and thus underscores the risks of \textit{human bias}, other evidence indicates that overrides can sometimes improve outcomes. For example, overrides have been shown to be beneficial when humans utilize domain-specific insights not captured by the algorithm \citep{kesavan2025profit,wang2025power}. Thus, despite significant advancements in automation, human operators can still possess critical \textit{private information} that enhances the quality of decision-making. These mixed findings raise an essential design question: \textit{How should override rights be structured to improve human-AI collaboration when AI systems operate autonomously?}

Designing effective override policies for autonomous AI systems raises two distinct challenges. First, override decisions are inherently endogenous. Evaluating their effectiveness requires a counterfactual comparison between the actual override and the algorithm’s original recommendation, yet only one of the two is observed. As a result, causal assessment typically relies on randomized experiments that induce exogenous variation in override behavior \citep{kesavan2020field, sun2022predicting} or on imputation using observational data \citep{ibanez2018discretionary, wang2025power}, which is more tractable when the outcome variables are binary. Second, override discretion is difficult to manage precisely because it is internal: workers must judge when intervention is warranted, despite cognitive limits \citep{fugener2022cognitive}, biased self-assessment \citep{de2023your}, and limited feedback \citep{kwon2022human}. These constraints make it difficult to assess, influence, or improve discretionary behavior directly.

Using inventory management as an example, we address these challenges and show that the causal value of human overrides depends critically on the structure of discretion. We introduce a \textit{constrained override} policy that limits the number of overrides workers can make per decision episode. This simple cardinality constraint changes the decision environment: when the cap binds, selecting one override requires forgoing another. We term the resulting prioritization the \textit{selective filtering mechanism}: workers use their information to retain the overrides they perceive as most valuable, which can improve average override quality.

Our study is motivated by a collaboration with Fengyi Technology, a leading smart vending operator in China that operated more than $180{,}000$ machines by the end of $2025$. Its centralized algorithms recommend replenishment quantities for each stock-keeping unit (SKU\endnote{An SKU denotes a distinct product type, whereas an item denotes a single unit of that product.}) at each machine, and local workers implement these recommendations. The firm's \textit{relative margin}, $r$, captures the tradeoff between sales margins and inventory holding costs, with the approximate range of $0.8$ to $0.95$ indicating that losing a sale is substantially more costly than holding an additional item. Upward overrides therefore remain unrestricted.

Workers receive sales bonuses and deductions for aged inventory, approximately aligning their incentives with the firm's objective. Nevertheless, a $2022$ pilot allowing unrestricted downward overrides, which reduce recommended replenishment quantities, led to sales losses. The company consequently retained a \textit{no downward override} baseline, with exceptions for operational necessities such as warehouse stockouts. Yet workers possess local information unavailable to the algorithm, including customer preferences and nearby competitor promotions. Fengyi therefore sought a scalable policy that could preserve the value of this information while limiting harmful overrides.

To formalize our proposed solution, we develop a stylized Bayesian model \citep{arrow1973education, spence1973job} in which workers form assessments of the potential benefit of overriding each SKU using centralized information and local information acquired through repeated service. These assessments are informative but noisy, and workers hold optimistic beliefs about the average value of downward overrides. Under a free override policy, workers act on every assessment they perceive as favorable, resulting in a mix of low- and high-quality decisions. A constrained override policy, by limiting the number of allowable overrides, induces endogenous selection based on signal strength: workers are forced to rank-order override opportunities and act only on the most compelling ones. This selective filtering process raises the average quality of overrides without requiring the ex-ante identification of useful cases. The model's informational premise is consistent with descriptive evidence from the field experiment introduced below. Under free override, the share of \textit{safe} overrides---those leaving positive inventory immediately before the next replenishment---is higher for workers with greater machine--SKU familiarity ($89.13\%$ versus $75.27\%$). Nevertheless, \textit{risky} overrides, which leave no inventory and may involve lost sales, persist even at high familiarity.

We designed a randomized field experiment involving $553$ workers and approximately $4,000$ SKUs. Workers were randomly assigned to one of three override policies that differed only in the extent of downward discretion: (1) \textit{no override}, which prohibited reductions to algorithmic recommendations; (2) \textit{free override}, which permitted unrestricted reductions; and (3) \textit{constrained override}, which limited each worker to reducing at most two SKUs per machine per replenishment (override cap $K=2$). All groups retained full discretion to increase replenishment quantities and, in rare cases, such as warehouse stockouts, to make additional reductions. A key feature of our design is that randomization occurs at the \textit{worker level}. Building on prior field experiments on human discretion \citep{kesavan2020field, liang2024examining}, we assign each worker an override policy while preserving the choice of which SKUs to reduce. This allows us to examine whether constrained override induces workers to prioritize their strongest override opportunities, a central implication of selective filtering.

We estimate the causal effects of the override policies using a difference-in-differences specification with a worker--machine--day panel, comparing the constrained and free override groups to the no-override baseline. We measure performance outcomes at the machine level based on SKU-level sales and inventory between consecutive replenishments. The free override policy reduces inventory by $1.95$ log points ($p<0.01$) but also lowers daily sales by $1.19$ log points and sales probability by $1.38$ percentage points ($p<0.01$). Constrained override instead reduces inventory by $1.28$ log points ($p<0.01$), with no statistically significant change in sales or sales probability. Valuing these outcomes at different relative margins, we find that free override has a positive profit point estimate at low $r$, when inventory savings dominate, but a negative estimate at high $r$, when lost sales are more costly. Across the company's approximate range of $r$ (between $0.8$ and $0.95$), constrained override has the larger profit point estimate.

To examine the outcomes of overrides induced by each policy, we use a local average treatment effect (LATE) framework \citep{imbens1994identification}, with randomized policy assignment as an instrument for SKU-level override decisions. Under the IV assumptions, the LATE captures the causal effect for instances whose override status changes with the policy; the two policy comparisons concern different complier populations. Comparing the constrained override policy with the no-override policy, we estimate a LATE of $-3.45$ items in inventory ($p<0.01$), with no statistically significant change in sales quantity or probability. In contrast, the free override policy yields smaller inventory reductions ($-1.26$ items) but significantly decreases the probability of sale by $10.55$ percentage points ($p<0.01$), a pattern consistent with selective filtering.

To understand the information behind workers' choices, we examine override selection and safety, distinguishing centralized information recorded by the firm from worker experience and machine--SKU familiarity, which proxy for private information. After controlling for centralized information, one additional monthly replenishment of the same SKU at the same machine is associated with a $0.8$-percentage-point increase in safe-override probability under constrained override, compared with $0.4$ percentage points under free override (difference $=0.4$ percentage points, $p<0.05$). This pattern is consistent with workers using knowledge of a particular product at a particular machine to select safer overrides.

Finally, we examine when constrained override is more effective and how firms should set the override cap $K$. Constrained override reduces inventory without statistically significant sales declines for top-sales and growth-stage SKUs, senior workers, and high-familiarity machine--SKU pairs; inventory reductions reach $1.59$ log points among senior workers and $2.07$ log points for high-familiarity pairs (both $p<0.01$). In simulations, higher caps generate diminishing additional inventory savings while sales losses accumulate. The simulated profit-maximizing cap therefore weakly decreases as $r$ rises, while greater worker experience and machine--SKU familiarity support higher caps near the company's relative margin. These findings suggest calibrating discretion to the firm's relative margin and workers' private information.

Before proceeding, we summarize our key contributions.
\begin{itemize}
    \item[1.] \textit{We provide causal field evidence on the value of human overrides in automated inventory replenishment.} We show how the extent of downward discretion changes the balance between inventory savings and lost sales. Worker-level randomization preserves this choice and identifies the effects of override policies over repeated replenishment cycles.

    \item[2.] \textit{We develop selective filtering as a mechanism through which constrained override can improve replenishment decisions.} When workers have informative but noisy assessments and optimistic beliefs about the value of overrides, a binding cap can improve override quality by inducing them to prioritize their strongest opportunities. This allows the firm to use local information without identifying useful overrides in advance. 

    \item[3.] \textit{We characterize the economic trade-off in setting the override cap.} Allowing more downward overrides can generate additional inventory savings but also greater lost-sales risk. Tightening the cap limits this risk but can exclude profitable reductions. Relative margin determines how the firm values these consequences. Our policy comparisons evaluate this trade-off, while simulations provide suggestive guidance on calibrating the cap to the firm's relative margin, worker experience, and machine--SKU familiarity.
\end{itemize}

\section{Literature Review}
\label{section: literature}
We relate our study to three streams of research: human discretion in operations, the design of human--AI collaboration, and behavioral explanations of human judgment.

Research on automated inventory management is closest to our study. Using retail ordering data, \citet{van2010ordering} show that managers incorporate operational considerations, including handling workload, that automated replenishment systems may omit. Related work establishes the value of downstream information for replenishment decisions \citep{cui2015information}. In a field experiment, \citet{kesavan2020field} find that allowing merchants to override product-removal recommendations lowers profit on average but benefits growth-stage products. Another field experiment finds that judgmental adjustments to forecast inputs can improve profitability \citep{kesavan2025profit}. Using a structural model and data on planners' decisions, \citet{stip2024reluctant} link planners' reluctance to reduce spare-parts base stocks to perceived shortage costs. Within vending operations, \citet{kawaguchi2021will} examines adherence to assortment recommendations in a field experiment; we study replenishment quantities. These studies show that the value of human intervention depends on what workers adjust, the information they possess, and the costs of excess inventory and shortages. Building on this evidence, we examine how limiting downward discretion affects replenishment performance.

Related studies examine human discretion in other operational settings. In pricing, observational studies examine local price adjustments \citep{phillips2015effectiveness} and adherence to markdown recommendations \citep{caro2023believing}; \citet{liang2024examining} combine a pricing field experiment with laboratory tests of behavioral explanations. In service operations, workers reorder queues or assign customers to tables \citep{ibanez2018discretionary,tan2020behavioral}. Observational studies examine judgmental forecasting and labor scheduling \citep{fildes2009effective,kwon2022human}, while a field experiment studies warehouse packing recommendations that account for workers' deviations \citep{sun2022predicting}. Recent customer-service experiments examine workers' use of AI-generated suggestions \citep{ni2026generative} and their supervision of autonomous AI, including taking over interactions \citep{wang2026agentic}. In our setting, a downward override directly changes the stock available to meet demand before the next visit. Measuring sales and inventory together is therefore essential for evaluating its economic consequences.

A second stream studies how firms can design human--AI collaboration through the advice workers receive, the information available to them, and the decisions they can make. Research on advice examines human and algorithmic curation \citep{peukert2024editor}, medical diagnosis \citep{lebovitz2022engage}, and recommendations adapted to expected adherence \citep{grand2024best}. Research on information examines disagreement in lending \citep{wang2025power}, elicitation of human judgment for prediction systems \citep{ibrahim2021eliciting}, and the roles of private information and information access in advice use and complementarity \citep{balakrishnan2025human,lu20251+}. Other work examines how to structure human decisions through forced intervention \citep{cao2024impact}, human--AI teamwork \citep{fugener2021will}, safeguards informed by human knowledge \citep{chen2022algorithmic}, and interpretable guidance for sequential decisions \citep{bastani2025improving}.

Limited discretion is particularly relevant to our policy design. In laboratory forecasting experiments, \citet{dietvorst2018overcoming} show that limited modifications can increase algorithm use and improve forecast accuracy; their study includes a limit on the number of forecasts participants can change. Delegation theory examines how a principal restricts the decisions available to an informed but biased agent \citep{alonso2008optimal}. We examine a cap on downward overrides in inventory replenishment, where algorithmic recommendations are implemented by default. The firm determines how many SKUs workers may reduce per machine replenishment, while workers choose which SKUs to override and by how much. We use relative margin to examine how the sales--inventory trade-off affects the choice of cap.

A third stream helps explain why workers with useful information can still make poor decisions and how the decision environment affects their judgment. Inventory models study bounded rationality \citep{su2008bounded}, while laboratory experiments link overprecision to biased orders \citep{ren2013overconfidence}. Related research examines overconfidence in competition, distribution, and transshipment \citep{li2017overconfident,li2019overconfident,li2021transshipment}, as well as incentives and strategic ordering behavior \citep{feng2017modeling,ovchinnikov2015compete,quiroga2019behavioral}. In human--AI collaboration, laboratory studies identify difficulties in delegation \citep{fugener2022cognitive} and examine how transparency, complexity, and system load affect advice use \citep{lehmann2022risk,snyder2025algorithm}. Experiments on algorithm use examine responses to observed errors and uncertainty \citep{dietvorst2015algorithm,dietvorst2020people}. A field experiment examines customers' responses to chatbot disclosure \citep{luo2019frontiers}. Using data from an algorithmic recommendation rollout, \citet{liu2023algorithm} examine drivers' reliance on contextual experience and peers' actions. Theory shows how selective outcome feedback can impede learning about algorithm quality \citep{de2023your}. The decision environment can also improve judgment: research examines how choice-set size affects decisions \citep{iyengar2000choice}, and field evidence from emergency departments shows that cognitive constraints can improve decisions when decision aids are available \citep{shanmugam2020decision}. Our model examines override selection when workers' assessments are informative but noisy and they hold optimistic beliefs about the value of overrides. Under the model's conditions, a binding cap can improve override quality by inducing workers to prioritize their strongest opportunities, with assessment precision held fixed. Our field evidence is consistent with this selective filtering mechanism.

\section{Research Context, Theoretical Rationale, and Experiment Design}
\label{section: experiment}
\subsection{Empirical Context}\label{subsection: empirical context}
Fengyi Technology is a leading operator of smart vending machines in China’s self-service retail market. By the end of $2025$, the company operated more than $180{,}000$ machines across more than $70$ metropolitan areas and served more than $80$ million consumers. A typical machine stocks $40$--$50$ ready-to-consume SKUs and is deployed in locations such as office buildings, hotel lobbies, and outdoor public areas (see \Cref{fig: smart vending machine} in Section~\ref{appendix: exp} of the online appendix). All machines are connected to a centralized server that records sales and inventory in real time and coordinates replenishment across the network.

\textbf{Replenishment Objective and Algorithm Design.}
For each SKU, the firm determines how much to replenish until the next service visit. Replenishing more protects sales but raises inventory-related costs; replenishing less reduces inventory but increases the risk of stockouts and lost sales. Following the standard sales--inventory trade-off in inventory replenishment \citep{kesavan2020field}, we can summarize this decision using the stylized profit function
\begin{equation}\label{equation: replenishment value}
    \mathrm{Profit}(q) = r\,\mathrm{Sales}(q) - (1-r)\,\mathrm{Inventory}(q), \qquad r \coloneqq \frac{\mathrm{Margin}} {\mathrm{Margin}+\mathrm{Holding\ Cost}} \in(0,1).
\end{equation}
Here, $q$ is the replenishment quantity, while $\mathrm{Sales}(q)$ and $\mathrm{Inventory}(q)$ are the resulting average daily sales and average inventory until the next replenishment. We refer to $r$ as the \textit{relative margin}. For operational simplicity, the company applies a common $r$ across products. Although proprietary, its value lies between approximately $0.8$ and $0.95$, placing greater weight on sales than inventory. This common setting is suitable for an assortment consisting primarily of packaged, ready-to-consume products with broadly similar replenishment cycles and economic characteristics. The company translates this trade-off into a common target service level, calibrated based on the relative margin $r$. Its centralized system first forecasts SKU-level demand using historical sales, current inventory, calendar variables, weather, and location characteristics. At each replenishment review, it uses the forecast and service level to determine a target inventory level and recommends the quantity needed to bring existing inventory up to that target, subject to machine capacity and warehouse availability. 

\textbf{Worker Incentives.}
The recommended replenishment is then implemented by a large network of workers. 
Workers typically manage approximately $80$--$100$ machines and service $10$--$20$ machines per day. Compensation consists of a base salary, a sales bonus, and a deduction for aged inventory:
\begin{equation*}
    \mathrm{Compensation} = \mathrm{Base} + \mathrm{Sales\ Bonus} - \mathrm{Aged\ Inventory\ Deduction} \approx \mathrm{Base} + \lambda\,\mathrm{Profit}, \qquad \lambda>0.
\end{equation*}
The proportional relation is a stylized approximation of the intended alignment between workers' incentives and the firm's sales--inventory objective. Higher realized sales increase the bonus. Products that remain beyond the allowable selling period must be removed, and part of the associated loss is deducted from the responsible worker's compensation. Workers remain accountable for these outcomes whether they follow or modify the recommendation: unsold inventory may incur an aged-inventory deduction even when the worker follows the algorithm, while replenishing too little may reduce sales and the corresponding bonus.

\textbf{Downward Override Policies.}
Because the replenishment system places substantially greater weight on preserving sales than reducing inventory, upward overrides remain unrestricted throughout. How to structure workers' downward override policy is less clear. Workers observe centralized information through a mobile application, including the recommended quantity, current inventory, recent sales, and warehouse availability for each SKU (see \Cref{fig: screenshots for different groups} in Section~\ref{appendix: exp} of the online appendix). Through repeated service, they also acquire private information unavailable to the centralized system, such as customer preferences, nearby competitor promotions, and temporary demand shifts \citep{fisher2018competition}. Believing in the potential benefit of incorporating workers' private information, Fengyi Technology piloted a \textit{free override} policy in early 2022, allowing unrestricted downward overrides. Although some overrides improved individual replenishment decisions based on anecdotal evidence, the company reported increased lost sales during the pilot. The company therefore had to keep a \textit{no downward override} policy as its baseline, under which workers may reduce recommendations only when operationally necessary (e.g., warehouse stockouts).

\subsection{Constrained Override Policy and Selective Filtering}
\label{subsection: mechanism}
Beyond common operational exceptions, no downward override prevents workers from acting on potentially valuable local information, whereas free override permits both informative and noisy interventions. Therefore, we propose a \textit{constrained override} policy under which a worker may reduce algorithmic recommendations for at most $K$ SKUs per machine. (The override cap $K$ applies to the number of overridden SKUs, not the size of each reduction.) 

As illustrated in \Cref{fig: mechanism}, constrained override lies between the two existing policies, but it may not simply produce a random subset of the overrides made under free override. Suppose workers possess valuable local information but assess the profit effect of each override imperfectly. Under free override, one override does not displace another, so workers need not rank the adjustments they perceive as beneficial, allowing riskier overrides to accompany safer ones. Under constrained override, once desired overrides exceed the override cap $K$, selecting one requires forgoing another. Because workers are accountable for realized sales and inventory outcomes, they have an incentive to reserve their limited overrides for the SKUs with the largest perceived profit gains. The cap therefore changes not only how many overrides are made but also which ones are selected. This prioritization can retain safer overrides excluded under no downward override while filtering out riskier overrides permitted under free override. We call this mechanism \textit{selective filtering}.

\begin{figure}[t]    
    \FIGURE
    {\includegraphics[width=0.55\textwidth]{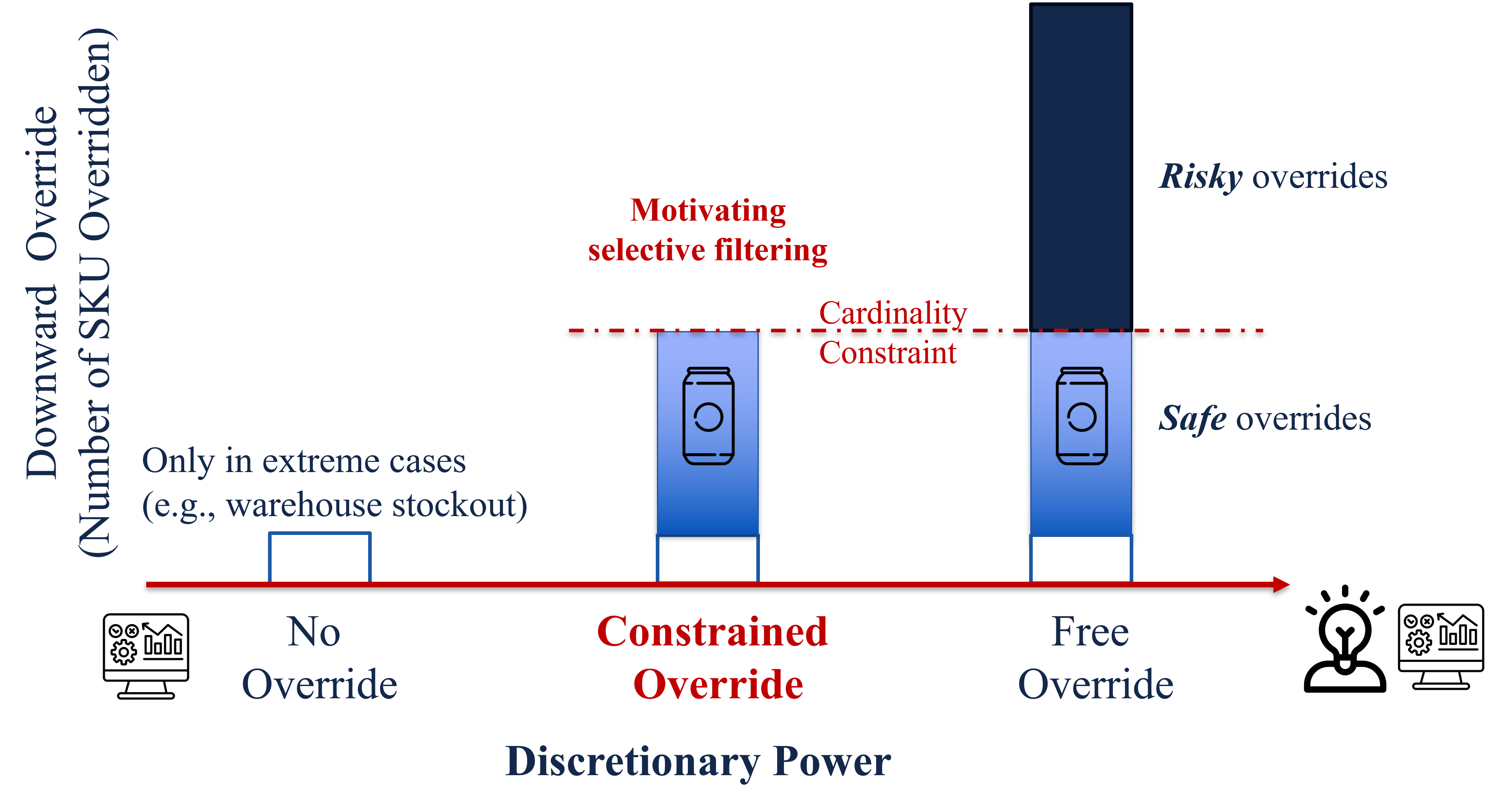}}
    {Illustration of the Selective Filtering Mechanism. \label{fig: mechanism}}
    {}
\end{figure}

As a preview of our experiment, we provide descriptive support for the premise underlying this mechanism: workers' assessments are informative but imperfect. We classify a downward override as \textit{safe} if the SKU retains positive inventory immediately before the next replenishment, and as \textit{risky} if no inventory remains and sales may have been lost. We use this classification as a proxy for override quality. Table~\ref{table:free_override_selective_filtering} reports safe-override shares by workers' prior experience with the same machine, SKU, and machine--SKU pair, as well as by worker tenure. Across all four measures, experienced workers have more downward-overridden SKUs per machine and safe-override shares that are $4$--$14$ percentage points higher. The largest difference is for machine--SKU familiarity, where the share increases from $75.27\%$ to $89.13\%$. Yet safe-override shares remain below $100\%$ even among more experienced workers. These patterns are consistent with workers having useful local information while facing uncertainty about individual override outcomes.

\newcommand{\tabindent}{\hspace*{1.5em}}
\begin{table}[!htbp]
    \linespread{0.5}
    \centering
    \scriptsize
    \caption{Safe Downward Overrides by Familiarity and Worker Seniority}
    \label{table:free_override_selective_filtering}
    \begin{tabular}{@{}p{0.55\linewidth}cc@{}}
        \toprule[1pt]
        \textit{Selection and outcome measure} & \textit{Low 50\%} & \textit{Top 50\%} \\
        \midrule
        \multicolumn{3}{@{}l}{\textbf{Panel A: Machine Familiarity}}\\[0.5ex]
        Monthly replenishments of machines by same worker (times/month)
        & $1.68$ [$1.45$] & $6.12$ [$3.13$] \\
        SKUs overridden per machine
        & $2.15$ [$3.51$] & $2.88$ [$4.03$] \\
        \tabindent SKUs with safe overrides per machine
        & $1.75$ [$2.99$] ($81.66\%$) & $2.50$ [$3.58$] ($86.98\%$) \\
        \midrule
        \multicolumn{3}{@{}l}{\textbf{Panel B: SKU Familiarity}}\\[0.5ex]
        Monthly replenishments of SKUs by same worker (times/month)
        & $9.51$ [$7.23$] & $21.68$ [$2.39$] \\
        SKUs overridden per machine
        & $0.69$ [$1.29$] & $1.99$ [$3.07$] \\
        \tabindent SKUs with safe overrides per machine
        & $0.52$ [$1.07$] ($76.16\%$) & $1.76$ [$2.77$] ($88.63\%$) \\
        \midrule
        \multicolumn{3}{@{}l}{\textbf{Panel C: Machine--SKU Familiarity}}\\[0.5ex]
        Monthly replenishments of machine--SKUs by same worker (times/month)
        & $0.05$ [$0.20$] & $2.60$ [$2.26$] \\
        SKUs overridden per machine
        & $0.74$ [$1.69$] & $2.11$ [$3.19$] \\
        \tabindent SKUs with safe overrides per machine
        & $0.56$ [$1.39$] ($75.27\%$) & $1.88$ [$2.90$] ($89.13\%$) \\
        \midrule
        \multicolumn{3}{@{}l}{\textbf{Panel D: Worker Seniority}}\\[0.5ex]
        Worker tenure (years)
        & $0.79$ [$0.39$] & $1.46$ [$0.09$] \\
        SKUs overridden per machine
        & $2.07$ [$3.52$] & $3.22$ [$4.14$] \\
        \tabindent SKUs with safe overrides per machine
        & $1.72$ [$2.99$] ($83.16\%$) & $2.80$ [$3.72$] ($87.13\%$) \\
        \bottomrule \\[-0.5ex]
        \multicolumn{3}{@{}p{0.98\linewidth}@{}}{%
            \textit{Note.} The sample consists of downward overrides made by workers in the free-override group during the treatment period. Familiarity measures are constructed from pre-treatment replenishment records. For a treatment-period downward override of SKU $k$ at machine $j$ by worker $i$, \textit{machine familiarity} is the monthly number of pre-treatment visits by worker $i$ to machine $j$; \textit{SKU familiarity} is the monthly number of pre-treatment replenishments of SKU $k$ by worker $i$ across all machines serviced by that worker; and \textit{machine--SKU familiarity} is the monthly number of pre-treatment replenishments of SKU $k$ at machine $j$ by worker $i$. \textit{Worker seniority} is the worker's tenure with the firm at the start of the experiment, measured in years. Low and Top subgroups are formed separately in each panel using a median split of the corresponding measure. Cells report means with standard deviations in brackets. Percentages in parentheses report safe overrides as a share of all downward overrides.
        } \\
    \end{tabular}
\end{table}

\textbf{Theoretical Model.}
To formalize this intuition, we develop a stylized Bayesian model to show how selective filtering can increase expected profit. Consider a worker servicing a machine with $N$ SKUs, indexed by $k\in[N]$. For SKU $k$, the algorithm recommends replenishment quantity $q_k^A$, and the worker may reduce it to $\widetilde q_k<q_k^A$. The incremental profit from this downward override, relative to following the recommendation, is
\begin{equation}\label{equation: override value}
    \eta_k \equiv \mathrm{Profit}(\widetilde q_k)-\mathrm{Profit}(q_k^A).
\end{equation}
An override is profitable if $\eta_k>0$ and harmful if $\eta_k<0$. For tractability, override values are i.i.d. with $\eta_k\sim\mathcal N(\mu,\sigma_\eta^2)$, where $\mu<0$. The negative mean captures average losses from candidate downward overrides, while variation across SKUs allows some overrides to be profitable.

The worker cannot observe $\eta_k$ when deciding: future demand is uncertain, and only one of $q_k^A$ and $\widetilde q_k$ can be implemented. The worker combines centralized information with local knowledge acquired through repeated service. Following a signal-based formulation \citep{arrow1973education,spence1973job}, we represent this assessment as $s_k=\eta_k+\epsilon_k$, where the assessment errors are i.i.d. with $\epsilon_k\sim\mathcal N(0,\sigma^2)$ and are independent of the override values. Thus, $s_k$ is informative about $\eta_k$ but remains noisy.

To capture the company's observation that workers tend to be optimistic about downward overrides, we assume that workers behave as if $\eta_k\sim\mathcal N(0,\sigma_\eta^2)$, although the true mean is $\mu<0$. Let $\mathbb E^W$ denote the worker's subjective expectation. Under this perceived prior, $\mathbb E^W[\eta_k\mid s_k]=\frac{\sigma_\eta^2}{\sigma_\eta^2+\sigma^2}s_k$. The worker therefore perceives an override as profitable if and only if $s_k>0$. Because perceived incremental profit is increasing in $s_k$, ranking overrides by perceived profit is equivalent to ranking them by their assessments.

To isolate selective filtering, we take each proposed reduction quantity as given and model which overrides the worker selects. Let $\mathcal A$ denote the collection of override sets permitted by the assigned policy. Given the incentive alignment described in Section~3.1, the worker chooses $A\in\mathcal A$ to maximize $\sum_{k\in A}\mathbb E^W[\eta_k\mid s_k]$. Under the \textit{no override} policy, abstracting from common operational exceptions, the worker follows all recommendations, so $A^{\mathrm{No}}=\emptyset$. Under the \textit{free override} policy, the worker overrides every SKU with a positive assessment, so $A^{\mathrm{Free}}=\{k:s_k>0\}$. Under the \textit{constrained override} policy, the worker implements at most $K$ overrides by selecting the SKUs with the highest positive assessments: $A^{\mathrm{Constrained}}=\{k:s_k>0,\ s_k\geq s_{(K)}\}$, where $s_{(1)}\geq\cdots\geq s_{(N)}$ are the ordered assessments. These decision rules formalize the three policies illustrated in \Cref{fig: mechanism}. For tractability, the following lemma and proposition are derived under the approximation $A^{\mathrm{Constrained}}\approx\{k:s_k\geq s_{(K)}\}$. This approximation is accurate when the probability of observing fewer than $K$ positive assessments is negligible.

\begin{lemma}[Selective Filtering] \label{lemma:conditional_positive}
    The expected incremental profit conditional on an override selected under the constrained policy is positive $\mathbb E\!\left[\eta_k\mid k\in A^{\mathrm{Constrained}}\right]>0$ if $a(K)>-\mu\tau/\sigma_\eta^2$, where $\tau^2\coloneqq\sigma_\eta^2+\sigma^2$ and $a(K)\coloneqq\frac{1}{K}\sum_{j=1}^{K}\mathbb E[U_{(j)}]$. Here, $U_{(1)}\geq\cdots\geq U_{(N)}$ are the order statistics of $U_k\coloneqq(s_k-\mu)/\tau$.
\end{lemma}

In \Cref{lemma:conditional_positive}, higher assessments correspond to higher true incremental profits on average. A binding cap exploits this relationship by retaining the highest $K$ assessments. The stated condition ensures that these selected overrides are profitable on average, despite the negative mean value of all override opportunities. The proof is in \Cref{appendix: bayesian}.
\begin{proposition}[Policy Comparison]\label{prop:filtering}
    The total expected incremental profits satisfy
    \begin{equation*}
        \mathbb E\!\left[\sum_{k\in A^{\mathrm{Constrained}}}\eta_k\right] > \mathbb E\!\left[\sum_{k\in A^{\mathrm{No}}}\eta_k\right] = 0 > \mathbb E\!\left[\sum_{k\in A^{\mathrm{Free}}}\eta_k\right]
    \end{equation*}
    if $h(-\mu/\tau)<-\mu\tau/\sigma_\eta^2<a(K)$, where $h(x)\coloneqq\phi(x)/[1-\Phi(x)]$, and $\phi$ and $\Phi$ are the standard normal density and distribution functions.
\end{proposition}
Under the condition in \Cref{prop:filtering}, constrained override increases expected profit relative to following the algorithm, whereas free override reduces it. This gain comes from selecting better overrides, with workers' information and assessment accuracy held fixed. We examine the total profit comparison at the machine level in \Cref{section: results}. We then study policy-induced overrides in \Cref{section:selective filtering} to assess the selection mechanism suggested by \Cref{lemma:conditional_positive}.

Prior research suggests that choice-set size and limited attention can also affect decision quality \citep{iyengar2000choice,iyengar2010choice,shanmugam2020decision,simon1955behavioral,sims2003implications}. To allow workers to choose how carefully to assess override opportunities, \Cref{appendix: alternative models} develops a rational-inattention model in which more accurate assessments are more costly. \edit{Under analogous conditions, the model yields the same ordering of expected incremental profit before information costs as \Cref{prop:filtering}. We also relax cross-SKU independence in \Cref{appendix: correlated signals}; allowing common components in override values and assessment errors preserves the ranking mechanism when workers' relative assessments remain informative, although stronger correlation in override values weakens the selection gain.}

Unlike interventions that improve case-by-case judgment through elicitation \citep{ibrahim2021eliciting}, transparency \citep{balakrishnan2025human}, interpretability \citep{lu20251+}, or forecast selection \citep{chen2023algorithm}, constrained override changes how workers select among override opportunities. The mechanism uses informative assessments to rank opportunities without requiring unbiased beliefs or precise estimates of every override's incremental profit. The policy therefore does not require worker retraining, redesign of the recommendation algorithm, or individualized feedback. The model also guides the choice of the experimental cap: under its assumptions, the profit-maximizing cap lies below the expected number of overrides under free override (\Cref{prop:optimal_K_ratio} in \Cref{appendix: bayesian}).

\subsection{Experiment Design}
\label{subsection: experiment design}
Having established the theoretical foundation of the policy and its mechanism, we designed and implemented a randomized field experiment. A central feature of our experimental design is randomization at the \textit{worker level}. We build on earlier field experiments on human discretion (e.g., \citealt{kesavan2020field, liang2024examining}) by assigning each worker an override policy while leaving workers to choose which SKUs to override. Under constrained override, this choice is made within a common cap on the number of downward adjustments per machine replenishment. The design allows us to test whether limiting discretion changes workers' allocation of overrides and the resulting sales and inventory outcomes.

The $46$-day experiment started on May $4$, $2023$, and included all full-time workers responsible for managing at least 10 vending machines. We randomized $553$ workers into three groups: $40\%$ ($n=221$) to the no override group, prohibiting downward override; $40\%$ ($n=223$) to the constrained override group, allowing downward overrides for at most $K$ SKUs per machine; and $20\%$ ($n=109$) to the free override group, with no restrictions on downward overrides. All groups could increase algorithm-suggested quantities. Guided by \Cref{prop:optimal_K_ratio}, which suggests setting the cap below the average to promote selective filtering, we use evidence from the early pilot to calibrate the constrained override policy. In the free-override pilot, among machines with at least one downward override, the median number of reductions was $4$ SKUs per machine. Based on this distribution and the company's business insight, we set $K=2$, which imposes a conservative cap.

We used pre-treatment data from April $1$ to May $3$, $2023$, to balance worker experience, machine and SKU characteristics, and prior override activity through repeated randomization. We conducted pairwise $t$-tests after each assignment and retained an assignment with no significant differences across these covariates at the $1\%$ level. Table~\ref{table: randomization check} in \Cref{appendix: exp} reports all balance measures. Moreover, we communicated assigned policies through banner notifications in the replenishment app (Figure~\ref{fig: screenshots for different groups} in \Cref{appendix: exp}). The constrained-group banner displays the two-SKU cap and current override count; the app prompts workers to reduce that count if they exceed the cap. The free-group banner states that reductions are unrestricted. To account for the presence of a banner itself, the no-override group received a placebo banner displaying a standard replenishment rule unrelated to the experiment.

\section{Impact of Constrained and Free Override Policies} \label{section: results}
We evaluate policy effects using a panel indexed by $(i,j,t)$, where each observation is worker $i$'s replenishment of machine $j$ on day $t$. We measure SKU-level outcomes between consecutive replenishments and average them across all SKU instances within each machine replenishment. Consistent with \Cref{equation: replenishment value}, \textit{Sales} is average daily quantity sold, and \textit{Inventory} is average quantity on hand during that window. \textit{Sales Probability} is the proportion of these SKUs with at least one unit sold during the window. It complements sales quantity by showing how broadly sales are distributed across these SKUs. Figure~\ref{fig: model free} in \Cref{appendix: main results} of the online appendix plots daily averages of the machine-level sales and inventory measures by policy group.

We estimate the following difference-in-differences (DID) specification:
\begin{equation}\label{eq:did_machine}
    \begin{aligned}
        \textit{Outcome}_{ijt} = \alpha &+ \delta_1 (\textit{ConstrainedOverride}_i \cdot \textit{Post}_t) + \delta_2  (\textit{FreeOverride}_i \cdot \textit{Post}_t) + \gamma_i + \delta_{c(j)t} + \varepsilon_{ijt}
    \end{aligned}
\end{equation}
where $\textit{ConstrainedOverride}_i$ and $\textit{FreeOverride}_i$ indicate assigned policies, with no override as the control, and $\textit{Post}_t=1$ on and after the implementation date. Worker fixed effects $\gamma_i$ absorb time-invariant worker differences, and city--date fixed effects $\delta_{c(j)t}$ absorb local shocks. The coefficients $\delta_1$ and $\delta_2$ identify policy effects relative to no override. Throughout the paper, we cluster standard errors by worker to account for serial correlation and the unit of randomization.

\begin{table}[!htbp] 
    \linespread{0.5}
    \centering 
    \scriptsize
    \caption{Causal Impact of Free and Constrained Override Policies (Machine-Level Estimates)} 
    \label{table: main results}
    \begin{tabular}{@{\extracolsep{5pt}}lccc} 
        \toprule[1pt]\\[-1.5ex]
        
         & $\log(\textit{Sales}+1)$ & \textit{Sales Probability} & $\log(\textit{Inventory}+1)$\\[0.5ex]

        \midrule
        
         $ConstrainedOverride_i \cdot Post_{t}$ & $0.0010$ & $-0.0020$ & $-0.0128^{***}$\\ 
         $ $ & $(0.0039)$ & $(0.0027)$ & $(0.0049)$\\ 
         $FreeOverride_i \cdot Post_{t}$ & $-0.0119^{***}$ & $-0.0138^{***}$ & $-0.0195^{***}$\\ 
         $ $ & $(0.0043)$ & $(0.0033)$ & $(0.0066)$\\ 
        
        \midrule
        
        Worker fixed effects & Yes & Yes & Yes\\
        (City-Date) fixed effect & Yes & Yes & Yes\\
        R$^{2}$ & $0.1751$ & $0.2072$ & $0.1708$\\
        Observations & $639,112$ & $639,112$ & $639,112$\\ 
        
        \bottomrule\\[-1ex]
        
        \multicolumn{4}{l}{\textit{Note.} Robust standard errors clustered by each worker are in parentheses.}\\
        \multicolumn{4}{l}{$^{*}p<0.1$; $^{**}p<0.05$; $^{***}p<0.01$}
    \end{tabular} 
\end{table}

Table~\ref{table: main results} presents the main policy effects. We report $100$ times the coefficients on $\log(1+\textit{Sales})$ and $\log(1+\textit{Inventory})$ as log points, and probability effects as percentage points. The free override policy reduces inventory by $1.95$ log points ($p<0.01$), but also lowers sales by $1.19$ log points ($p<0.01$) and sales probability by $1.38$ percentage points ($p<0.01$). This pattern highlights the operational cost of unrestricted discretion: inventory reductions are accompanied by lost sales in this case, consistent with prior evidence that unstructured human overrides can degrade algorithmic performance \citep{kesavan2020field,liang2024examining,wang2025power}. The constrained override policy reduces inventory by $1.28$ log points ($p<0.01$), with no statistically significant change in sales or sales probability.

\begin{table}[!htbp]
    \linespread{0.5}
    \centering
    \scriptsize
    \caption{Effects of Override Policies on Replenishment and Override Behavior}
    \label{table:effect_on_behavior}
    \begin{tabular}{lccccccc}
        \toprule
        \multicolumn{8}{l}{\textbf{Panel A: Replenishment Behavior}} \\
        \midrule
        & \multicolumn{3}{c}{\textit{Algorithm Suggested}} & \multicolumn{3}{c}{\textit{Actual Behavior}} & \\
        \cmidrule(lr){2-4} \cmidrule(lr){5-7}
         & Qty & SKU & Qty/SKU & Qty & SKU & Qty/SKU & \\
         & (Item) & (Count) & (Item) & (Item) & (Count) & (Item) & \\
        \midrule
        $ConstrainedOverride_i \cdot Post_t$
        & $0.7625$ & $-0.0186$ & $0.0502$
        & $-0.7325$ & $-0.1104$ & $-0.1083^{**}$ & \\
        & ($0.9512)$ & ($0.1776$) & ($0.0357$)
        & ($0.8487$) & ($0.1301$) & ($0.0441$) & \\
        
        $FreeOverride_i \cdot Post_t$
        & $3.6673^{**}$ & $0.2669$ & $0.2244^{***}$
        & $-3.8175^{***}$ & $-0.5505^{***}$ & $-0.1816^{**}$ & \\
        & ($1.6838$) & ($0.2965$) & ($0.0684$)
        & ($1.0559$) & ($0.1674$) & ($0.0736$) & \\
        
        \midrule
        Worker \& City-Date FE & Yes & Yes & Yes & Yes & Yes & Yes & \\
        R$^2$ & $0.1937$ & $0.2099$ & $0.2105$ & $0.1521$ & $0.1431$ & $0.1112$ & \\
        Observations & $639{,}112$ & $639{,}112$ & $639{,}112$ & $639{,}112$ & $639{,}112$ & $639{,}112$ & \\[0.5ex]
        
        \toprule
        \multicolumn{8}{l}{\textbf{Panel B: Override Behavior}} \\
        \midrule
        & \multicolumn{4}{c}{\textit{Downward Override}} & \multicolumn{3}{c}{\textit{Upward Override}} \\
        \cmidrule(lr){2-5} \cmidrule(lr){6-8}
         & SKU & Qty/SKU & Zero Share & Qty Red. & SKU & Qty/SKU & Qty Inc. \\
         & (Count) & (Item) & (Share) & (Item) & (Count) & (Item) & (Item) \\
        \midrule
        $ConstrainedOverride_i \cdot Post_t$
        & $0.3624^{***}$ & $0.7542^{***}$ & $0.0109^{***}$ & $2.0719^{***}$
        & $0.2296^{**}$ & $-0.0403$ & $0.5769$ \\
        & ($0.1215$) & ($0.1260$) & ($0.0028$) & ($0.3755$)
        & ($0.1168$) & ($0.0794$) & ($0.4696$) \\
        
        $FreeOverride_i \cdot Post_t$
        & $1.9564^{***}$ & $0.1201$ & $0.0267^{***}$ & $6.6767^{***}$
        & $0.0777$ & $-0.1810^{**}$ & $-0.8080$ \\
        & ($0.2331$) & ($0.1195$) & ($0.0052$) & ($0.7933$)
        & ($0.2179$) & ($0.0706$) & ($0.8577$) \\
        
        \midrule
        Worker \& City-Date FE & Yes & Yes & Yes & Yes & Yes & Yes & Yes \\
        R$^2$ & $0.4334$ & $0.1492$ & $0.3371$ & $0.3633$ & $0.3885$ & $0.2219$ & $0.2969$ \\
        Observations 
        & $639{,}112$ & $192{,}300$ & $639{,}112$ & $639{,}112$
        & $639{,}112$ & $344{,}391$ & $639{,}112$ \\
        \bottomrule \\[-0.5ex]
        \multicolumn{8}{p{0.97\textwidth}}{\textit{Note.} Qty Red./Inc. denotes total quantity reduced/increased per machine; Zero Share is the share of SKUs with replenishment reduced to zero. In Panel B, Qty/SKU is the quantity change per overridden SKU. Robust standard errors clustered at the worker level are in parentheses.} \\
        \multicolumn{8}{l}{$^{*}p<0.1$; $^{**}p<0.05$; $^{***}p<0.01$}
    \end{tabular}
\end{table}

\Cref{table:effect_on_behavior} applies the same DID specification to replenishment and override decisions, helping explain how workers achieve these outcomes. In Panel A, free override reduces total replenishment by $3.82$ items per machine ($p<0.01$), while the constrained-policy effect is not statistically significant. Panel B shows that constrained and free override increase the number of downward-overridden SKUs by $0.36$ and $1.96$, respectively, relative to no override (both $p<0.01$). The reduction per selected SKU increases by $0.75$ items under constrained override ($p<0.01$); the $0.12$-item increase under free override is not statistically significant. The cap thus leaves room for substantial adjustments on selected SKUs. This pattern is consistent with workers changing which SKUs they reduce, rather than simply making every adjustment smaller.

\begin{figure}[htbp]
    \FIGURE
    {\includegraphics[width=0.75\textwidth]{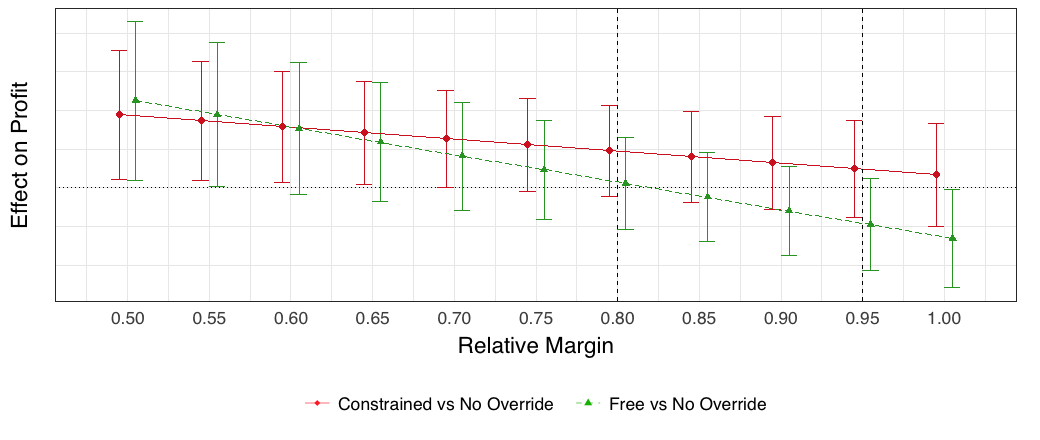}}
    {Profits Across Relative Margins by Override Group \label{fig: profit}}
    {Error bars indicate $95\%$ confidence intervals based on standard errors clustered at the worker level. Numerical labels on the vertical axis are omitted to protect confidential company information.}
\end{figure}

To evaluate sales and inventory jointly, \Cref{fig: profit} reports each policy's estimated profit effect relative to no downward override across relative margins $r$, using the profit measure in \Cref{equation: replenishment value} and the specification in \eqref{eq:did_machine}. At low $r$, both policies have positive profit point estimates because inventory savings receive greater weight. As $r$ rises, free override's estimate falls sharply because its sales losses become more costly, whereas constrained override's estimate declines more slowly because its sales effect is close to zero. Within the company's approximate range of $r$ ($0.8$--$0.95$), constrained override's profit point estimates remain positive, whereas free override's become negative. More importantly, if the cap merely reduced inventory savings and sales losses proportionally, its expected profit effect would move toward zero but remain negative wherever free override's effect is negative. The positive constrained-override point estimates in that region therefore suggest an improved sales--inventory trade-off, consistent with workers selecting better overrides.

We report a series of robustness checks in the online appendix. 
The main sales--inventory findings remain similar when we use level outcomes or exclude weekends and major holidays (\Cref{appendix: main result}). Because earlier decisions can affect subsequent recommendations and workers may learn from their outcomes, we also examine algorithmic recommendations and human learning in \Cref{subsection:algorithm_suggestions_learning}. Controlling for the implied inventory target leaves the main results largely unchanged. Workers adjust to previous outcomes, but weekly DID estimates show that constrained override's inventory benefit appears in the first week and does not build steadily thereafter. These findings are consistent with better override selection from the beginning, alongside learning across replenishments.

\section{\edit{Evidence on Selective Filtering}}
\label{section:selective filtering}
We next examine whether the observed patterns are consistent with our proposed mechanism and assess whether better override selection helps explain the policy effects above (\Cref{lemma:conditional_positive}).

\subsection{Local Average Treatment Effect (LATE)}
Workers choose which SKUs to override, so observed outcome differences between overridden and non-overridden SKUs may reflect the conditions that prompted those choices. We therefore use randomized policy assignment as an instrument to estimate the local average treatment effect (LATE) of downward overrides \citep{imbens1994identification}. We compare constrained and free override separately with no override.

We call a SKU-replenishment instance a \textit{complier} if the worker would reduce the recommendation under the active policy but follow it under no override. Let $Y_{ijkt}(d)$ denote the potential outcome when the recommendation is reduced ($d=1$) or followed ($d=0$) \citep{rubin1974estimating}. Under the instrumental-variable assumptions, the LATE is the average effect of overriding for these compliers, $\text{LATE}=\mathbb E[Y_{ijkt}(1)-Y_{ijkt}(0)\mid\text{Complier}]$. We defer a detailed illustration of this identification strategy to \Cref{appendix:selective filtering test} of the appendix.  \Cref{appendix:exclusion} discusses the identifying assumptions, potential violations of the exclusion restriction \citep{chen2024new}, and related empirical checks. We estimate the LATE using a DID framework within a two-stage least squares (2SLS) specification at the worker--machine--SKU replenishment-instance level $(i,j,k,t)$. The sample includes instances in which workers follow or reduce the recommendation. The two stages are:

\noindent\textit{First stage (override decision):}
\begin{equation}
    \textit{DownOverride}_{ijkt} = \pi_0 + \pi_1 (\textit{ConstrainedOverride}_i \cdot \textit{Post}_t) + \gamma_{ij} + \gamma_k + \delta_{c(j)t} + \varepsilon_{ijkt},
\end{equation}
\noindent\textit{Second stage (outcome consequence):}
\begin{equation}
    Y_{ijkt} = \alpha_0 + \beta_C \cdot \widehat{\textit{DownOverride}}_{ijkt} + \gamma_{ij} + \gamma_k + \delta_{c(j)t} + \varepsilon_{ijkt}.
\end{equation} 
Both stages include worker--machine ($\gamma_{ij}$), SKU ($\gamma_k$), and city--date ($\delta_{c(j)t}$) fixed effects, with worker-clustered standard errors. We use the same outcome definitions as in the main analysis at the SKU-replenishment level: average daily sales, average inventory until the next replenishment, and an indicator for any sale during that window. We use levels to express the complier effects in quantity or probability units.

\begin{table}[!htbp]
    \linespread{0.5}
    \centering
    \scriptsize
    \setlength{\tabcolsep}{5pt}
    \renewcommand{\arraystretch}{1}
    \caption{LATE with Downward Override}
    \label{table:late down}
    \begin{tabular}{@{}lcccc@{}}
        \toprule[1pt]
        & \multicolumn{1}{c}{\textit{First Stage}}
        & \multicolumn{3}{c}{\textit{Second Stage (LATE)}}\\
        \cmidrule(lr){2-2}\cmidrule(lr){3-5}
        & \makecell{\textit{Downward}\\\textit{Override}}
        & \makecell{\textit{Sales}\\\textit{(Unit: Item)}}
        & \textit{Sales Probability}
        & \makecell{\textit{Inventory}\\\textit{(Unit: Item)}}\\
        \midrule

        \multicolumn{5}{@{}l}{%
            \textbf{\textit{Panel A. Constrained vs.\ No Override}}
        }\\[0.25ex]

        $ConstrainedOverride_i \times Post_t$
        & $0.0287^{***}$
        & 
        & 
        & \\[-0.4ex]
        
        &
        $(0.0068)$
        &
        &
        & \\[0.25ex]

        $\widehat{DownOverride}_{ijkt}$
        &
        & $0.2825$
        & $-0.0442$
        & $-3.4466^{***}$\\[-0.4ex]
        
        &
        & $(0.3331)$
        & $(0.0882)$
        & $(0.9058)$\\[0.25ex]

        (Worker-Machine), SKU, (City-Date) FE
        & Yes
        & Yes
        & Yes
        & Yes \\

        R$^2$
        & $0.3321$
        & $0.2914$
        & $0.3699$
        & $0.3025$\\

        First-stage F-statistic
        & $17.6790^{***}$
        &
        &
        & \\

        Observations
        & \multicolumn{4}{c}{$7{,}012{,}190$}\\

        \midrule

        \multicolumn{5}{@{}l}{%
            \textbf{\textit{Panel B. Free vs.\ No Override}}
        }\\[0.25ex]

        $FreeOverride_i \times Post_t$
        & $0.1292^{***}$
        &
        &
        & \\[-0.4ex]
        
        &
        $(0.0140)$
        &
        &
        & \\[0.25ex]

        $\widehat{DownOverride}_{ijkt}$
        &
        & $-0.1188$
        & $-0.1055^{***}$
        & $-1.2575^{***}$\\[-0.4ex]
        
        &
        & $(0.0761)$
        & $(0.0206)$
        & $(0.2509)$\\[0.25ex]

        (Worker-Machine), SKU, (City-Date) FE
        & Yes
        & Yes
        & Yes
        & Yes \\

        R$^2$
        & $0.3565$
        & $0.2858$
        & $0.3641$
        & $0.3201$\\

        First-stage F-statistic
        & $85.0620^{***}$
        &
        &
        & \\

        Observations
        & \multicolumn{4}{c}{$5{,}189{,}189$}\\

        \bottomrule
    \end{tabular}
    \par\smallskip
    \begin{minipage}{0.7\linewidth}
        \scriptsize
        \linespread{0.5}\selectfont
        \textit{Note.} Robust standard errors clustered at the worker
        level are in parentheses.\\
        $^{*}p<0.10$; $^{**}p<0.05$; $^{***}p<0.01$.
    \end{minipage}
\end{table}

Table~\ref{table:late down} estimates the LATE of downward overrides under both the constrained and free override policies. Relative to the no override group, which exhibits minimal downward overrides, the constrained policy increases override probability by $2.87$ percentage points ($F = 17.6790$), confirming the relevance of the instrument. The second stage shows that these additional overrides reduce inventory by an average of $3.45$ items ($p < 0.01$) with no statistically significant estimated change in sales quantity or sales probability. In contrast, the free override policy increases override probability by $12.92$ percentage points ($F = 85.0620$), but these overrides lead to less inventory reduction ($1.26$ items, $p < 0.01$) and a statistically significant decrease of $10.55$ percentage points in the probability of sale ($p < 0.01$).

The LATE estimates are largely consistent with selective filtering: constrained override induces decisions that combine inventory savings with no statistically detectable sales loss, while those induced by free override reduce sales probability. The estimated effects apply to different complier decisions, with the SKUs and reduction quantities selected under each policy. Complementary checks in \Cref{appendix:late-robustness} of the online appendix further report the corresponding reduced-form policy estimates and descriptive comparisons among realized upward overrides.

\subsection{Information Associated with Override Selection and Safety}
\label{subsection:information_filter}
To shed light on workers' private information, we examine how prior experience relates to override selection and safety after controlling for centralized information, and whether these relationships differ across policies. We distinguish centralized information available to the firm and workers, $X_{ijkt}$, from worker tenure, machine familiarity, and machine--SKU familiarity, $L_{ijkt}$. The experience measures capture opportunities to acquire general knowledge, knowledge of a location, and knowledge of a particular product at that location, respectively. Centralized information includes weekly sales, current inventory, the algorithmic recommendation, discount status, sale price, the firm's danger flag for SKUs at risk of removal because of low sales, and indicators for growth-stage SKUs, top-sales SKUs, and top-sales machines, including the interaction between the latter two. The experience measures are constructed before the experiment. 

Using treatment-period observations from the constrained and free override groups, we estimate two linear models:
\begin{equation}\label{equation:information_selection}
    \begin{aligned}
        \textit{DownOverride}_{ijkt} ={}& \alpha +\tau\,\textit{Constrained}_i +X_{ijkt}'\beta_X +L_{ijkt}'\beta_L \\
        &+ \left(\textit{Constrained}_i\times X_{ijkt}\right)'\theta_X + \left(\textit{Constrained}_i\times L_{ijkt}\right)'\theta_L +\delta_{c(j)t} +\varepsilon_{ijkt},
    \end{aligned}
\end{equation}
\begin{equation}\label{equation:information_quality}
    \begin{aligned}
        \textit{SafeOverride}_{ijkt} ={}& \alpha^{S} +\tau^{S}\textit{Constrained}_i +X_{ijkt}'\beta_X^{S} +L_{ijkt}'\beta_L^{S} \\
        &+ \left(\textit{Constrained}_i\times X_{ijkt}\right)'\theta_X^{S} + \left(\textit{Constrained}_i\times L_{ijkt}\right)'\theta_L^{S} +\delta_{c(j)t} +\nu_{ijkt}.
    \end{aligned}
\end{equation}
Equation~\eqref{equation:information_selection} models whether a downward override occurs. Equation~\eqref{equation:information_quality} uses realized downward overrides, with $\textit{SafeOverride}_{ijkt}=1$ when inventory remains positive before the next replenishment. In both models, $\beta$ describes the conditional associations under free override, $\beta+\theta$ gives those under constrained override, and $\theta$ measures the differences across policies; superscript $S$ denotes the safety coefficients. The policy coefficient $\tau$ is the constrained--free difference at the reference covariate values.

\begin{table}[!htbp]
    \linespread{0.5}
    \centering
    \scriptsize
    \setlength{\tabcolsep}{2.5pt}
    \renewcommand{\arraystretch}{1.02}
    
    \caption{Information and Experience Associated with Downward-Override Selection and Safety}
    \label{table:information_selection_quality}
    
    \begin{tabular}{
        @{}
        >{\raggedright\arraybackslash}p{0.30\textwidth}
        *{4}{>{\centering\arraybackslash}p{0.155\textwidth}}
        @{}
    }
        \toprule[1pt]
    
        &
        \multicolumn{2}{c}{\textit{Downward Override}}
        &
        \multicolumn{2}{c}{\textit{Safe Override}}\\
    
        \cmidrule(lr){2-3}
        \cmidrule(lr){4-5}
    
        &
        \multicolumn{2}{c}{
            \makecell{
                 $\widehat{\tau}=0.001$\\
                $(0.029)$
            }
        }
        &
        \multicolumn{2}{c}{
            \makecell{
                $\widehat{\tau}^{S}=-0.116^{**}$\\
                $(0.049)$
            }
        }\\
    
        \cmidrule(lr){2-3}
        \cmidrule(lr){4-5}
    
        \textit{Variable}
        & \makecell{\textit{Free}\\($\widehat{\beta}$)}
        & \makecell{\textit{Constrained} $-$ \textit{Free}\\($\widehat{\theta}$)}
        & \makecell{\textit{Free}\\($\widehat{\beta}^{S}$)}
        & \makecell{\textit{Constrained} $-$ \textit{Free}\\($\widehat{\theta}^{S}$)}\\
    
        \midrule
    
        \multicolumn{5}{@{}l}{\textbf{Centralized Information ($X_{ijkt}$)}}\\
    
        Weekly Sales
        & \makecell{$-0.016^{***}$\\$(0.002)$}
        & \makecell{$0.010^{***}$\\$(0.002)$}
        & \makecell{$-0.005^{***}$\\$(0.001)$}
        & \makecell{$0.004^{**}$\\$(0.002)$}\\
    
        Current Inventory
        & \makecell{$-0.020^{***}$\\$(0.002)$}
        & \makecell{$0.008^{***}$\\$(0.002)$}
        & \makecell{$0.007^{***}$\\$(0.001)$}
        & \makecell{$0.002$\\$(0.001)$}\\
    
        Recommended Quantity
        & \makecell{$0.039^{***}$\\$(0.002)$}
        & \makecell{$-0.015^{***}$\\$(0.002)$}
        & \makecell{$-0.002^{***}$\\$(0.000)$}
        & \makecell{$-0.001$\\$(0.001)$}\\
    
        Discounted
        & \makecell{$0.019^{***}$\\$(0.005)$}
        & \makecell{$-0.009$\\$(0.006)$}
        & \makecell{$0.037^{***}$\\$(0.005)$}
        & \makecell{$0.036^{**}$\\$(0.018)$}\\
    
        Sale Price
        & \makecell{$0.006$\\$(0.005)$}
        & \makecell{$-0.003^{***}$\\$(0.001)$}
        & \makecell{$-0.006$\\$(0.007)$}
        & \makecell{$0.002$\\$(0.001)$}\\
    
        Danger Flag
        & \makecell{$0.017^{***}$\\$(0.005)$}
        & \makecell{$-0.009^{*}$\\$(0.005)$}
        & \makecell{$0.036^{***}$\\$(0.006)$}
        & \makecell{$0.038^{**}$\\$(0.019)$}\\
    
        Growth-Stage SKU
        & \makecell{$-0.020^{***}$\\$(0.005)$}
        & \makecell{$0.012^{**}$\\$(0.006)$}
        & \makecell{$0.039^{***}$\\$(0.009)$}
        & \makecell{$0.003$\\$(0.011)$}\\
    
        Top-Sales SKU $\times$ Top-Sales Machine
        & \makecell{$-0.064^{***}$\\$(0.007)$}
        & \makecell{$0.032^{***}$\\$(0.009)$}
        & \makecell{$0.050^{*}$\\$(0.026)$}
        & \makecell{$-0.030$\\$(0.040)$}\\
    
        Top-Sales SKU
        & \makecell{$0.096^{***}$\\$(0.011)$}
        & \makecell{$-0.068^{***}$\\$(0.013)$}
        & \makecell{$0.089^{***}$\\$(0.018)$}
        & \makecell{$0.112^{***}$\\$(0.041)$}\\
    
        Top-Sales Machine
        & \makecell{$-0.044^{***}$\\$(0.012)$}
        & \makecell{$0.020$\\$(0.015)$}
        & \makecell{$-0.050^{**}$\\$(0.025)$}
        & \makecell{$0.031$\\$(0.039)$}\\
    
        \addlinespace[0.3ex]
        \multicolumn{5}{@{}l}{\textbf{Experience Proxies ($L_{ijkt}$)}}\\

        Worker Tenure (years)
        & \makecell{$0.084^{***}$\\$(0.023)$}
        & \makecell{$-0.053^{*}$\\$(0.027)$}
        & \makecell{$0.001$\\$(0.018)$}
        & \makecell{$-0.019$\\$(0.035)$}\\
    
        Machine Familiarity (visits per month)
        & \makecell{$-0.003$\\$(0.002)$}
        & \makecell{$0.003$\\$(0.002)$}
        & \makecell{$-0.003^{**}$\\$(0.002)$}
        & \makecell{$-0.002$\\$(0.002)$}\\
    
        Machine--SKU Familiarity (replenishment per month)
        & \makecell{$0.018^{***}$\\$(0.002)$}
        & \makecell{$-0.008^{***}$\\$(0.002)$}
        & \makecell{$0.004^{***}$\\$(0.001)$}
        & \makecell{$0.004^{**}$\\$(0.002)$}\\
    
        \midrule
    
        City-Date Fixed Effects
        & \multicolumn{2}{c}{Yes}
        & \multicolumn{2}{c}{Yes}\\
    
        R$^{2}$
        & \multicolumn{2}{c}{$0.1972$}
        & \multicolumn{2}{c}{$0.3181$}\\
    
        Observations
        & \multicolumn{2}{c}{$3{,}916{,}684$}
        & \multicolumn{2}{c}{$317{,}975$}\\
    
        \bottomrule
    \end{tabular}
    
    \vspace{1pt}
    
    \begin{minipage}{0.98\textwidth}
        \scriptsize
        \raggedright
        \textit{Note.} Danger Flag is a firm-assigned indicator for SKUs at risk of removal because of low sales. Robust standard errors clustered at the worker level are in parentheses. $^{*}p<0.10$, $^{**}p<0.05$, and $^{***}p<0.01$.
    \end{minipage}
\end{table}

\Cref{table:information_selection_quality} reports the associations of centralized information and prior experience with override selection and safety under constrained and free override. These results offer three insights into selective filtering. First, machine–SKU familiarity has a stronger association with safety under constrained than under free override. One additional monthly pre-treatment replenishment of the same SKU at the same machine is associated with a $0.4$-percentage-point increase in safe-override probability under free override and a $0.8$-percentage-point increase under constrained override (difference $=0.4$ percentage points, $p<0.05$), after controlling for centralized information and broader experience. The corresponding differences for worker tenure and general machine familiarity are not statistically significant. This pattern is consistent with useful private information arising from repeated experience with a particular product at a particular location, and with the cap helping workers use that knowledge to select safer overrides.

Second, the selection and safety results for top-sales SKUs are consistent with workers being more selective within this product category under the cap. For machines outside the top-sales category, top-sales SKU status is associated with a $9.6$-percentage-point increase in override probability under free override and a $2.8$-percentage-point increase under constrained override (difference $=-6.8$ percentage points, $p<0.01$). Among selected overrides in these machines, its association with safe-override probability is $8.9$ percentage points under free override and $20.1$ under constrained override (difference $=11.2$ percentage points, $p<0.01$).

Third, two other centralized information measures have stronger positive associations with safety under constrained than under free override. Among selected overrides, the association of the firm's danger flag with safe-override probability is $3.6$ percentage points under free override and $7.4$ under constrained override (difference $=3.8$ percentage points, $p<0.05$). Discounted status shows a similar pattern, with associations of $3.7$ and $7.3$ percentage points, respectively (difference $=3.6$ percentage points, $p<0.05$).

\section{Practical Design of Constrained Override Policies} \label{section:betteroverride}
Having examined the evidence on selective filtering, we turn to two practical design questions: \textit{When is constrained override most effective, and how many overrides should firms allow?} We address these questions through heterogeneous policy effects across worker and SKU characteristics and simulations under alternative values of the override cap $K$.

\subsection{When Is Constrained Override More Effective?} \label{subsection: HTE}
We examine how policy effects vary with the value of preserving sales and the information available for replenishment. Top-sales SKUs account for more sales, growth-stage SKUs have shorter demand histories, and worker tenure and machine--SKU familiarity capture opportunities to acquire local knowledge. \Cref{table: HTE} reports machine-level policy effects within these subgroups, with outcomes averaged across the relevant SKUs at each machine replenishment. Constrained override reduces inventory for top-sales SKUs ($-1.90$ log points, $p<0.01$), growth-stage SKUs ($-1.15$ log points, $p<0.05$), senior workers ($-1.59$ log points, $p<0.01$), and high-familiarity machine--SKU pairs ($-2.07$ log points, $p<0.01$), with no statistically significant sales declines in these groups. For low-sales SKUs, sales fall ($-0.56$ log points, $p<0.05$) while the inventory estimate is not significant. For mature-stage SKUs, both inventory ($-2.05$ log points) and sales probability ($-0.76$ percentage points) fall ($p<0.10$). None of the constrained-policy effects is statistically significant for junior workers, while inventory increases for low-familiarity pairs ($1.08$ log points, $p<0.10$). Free override has negative sales or sales-probability point estimates in every subgroup, although significance varies.

\begin{table}[!htbp]
    \centering
    \caption{Heterogeneous Policy Effects across SKU and Worker Characteristics}
    \label{table: HTE}

    \begingroup
    \scriptsize
    \setlength{\tabcolsep}{2.8pt}
    \renewcommand{\arraystretch}{0.90}

    \newcommand{\HTEest}[2]{%
        \shortstack{$#1$\\[-0.35ex]{\tiny $(#2)$}}}

    \newcommand{\HTEstat}[2]{%
        \shortstack{$#1$\\[-0.35ex]{\tiny $#2$}}}

    \begin{tabular}{@{}lcccccc@{}}

    \toprule[1pt]

    \textbf{Panel A}
    & \multicolumn{3}{c}{\textbf{Top-Sales SKU}}
    & \multicolumn{3}{c}{\textbf{Low-Sales SKU}}\\
    \cmidrule(lr){2-4}\cmidrule(l){5-7}

    & $\log(\textit{Sales}+1)$

    & \textit{Sales Probability}

    & $\log(\textit{Inventory}+1)$

    & $\log(\textit{Sales}+1)$

    & \textit{Sales Probability}

    & $\log(\textit{Inventory}+1)$
    \\

    \midrule

    \shortstack[l]{$ConstrainedOverride_i \cdot$\\[-0.35ex]$Post_t$}
    & \HTEest{0.0017}{0.0036}
    & \HTEest{-0.0029}{0.0019}
    & \HTEest{-0.0190^{***}}{0.0056}
    & \HTEest{-0.0056^{**}}{0.0028}
    & \HTEest{-0.0044}{0.0031}
    & \HTEest{-0.0014}{0.0047}\\[0.3ex]

    \shortstack[l]{$FreeOverride_i \cdot$\\[-0.35ex]$Post_t$}
    & \HTEest{-0.0158^{***}}{0.0041}
    & \HTEest{-0.0161^{***}}{0.0025}
    & \HTEest{-0.0340^{***}}{0.0077}
    & \HTEest{-0.0066^{**}}{0.0033}
    & \HTEest{-0.0126^{***}}{0.0037}
    & \HTEest{0.0116^{*}}{0.0060}\\[0.3ex]

    \shortstack[l]{$R^2$\\[-0.35ex]Observations}
    & \HTEstat{0.1219}{576{,}123}
    & \HTEstat{0.0948}{576{,}123}
    & \HTEstat{0.1474}{576{,}123}
    & \HTEstat{0.3410}{624{,}627}
    & \HTEstat{0.4260}{624{,}627}
    & \HTEstat{0.4113}{624{,}627}\\

    \midrule

    \textbf{Panel B}
    & \multicolumn{3}{c}{\textbf{Growth-Stage SKU}}
    & \multicolumn{3}{c}{\textbf{Mature-Stage SKU}}\\
    \cmidrule(lr){2-4}\cmidrule(l){5-7}

    \shortstack[l]{$ConstrainedOverride_i \cdot$\\[-0.35ex]$Post_t$}
    & \HTEest{0.0021}{0.0040}
    & \HTEest{-0.0031}{0.0029}
    & \HTEest{-0.0115^{**}}{0.0054}
    & \HTEest{-0.0020}{0.0084}
    & \HTEest{-0.0076^{*}}{0.0040}
    & \HTEest{-0.0205^{*}}{0.0116}\\[0.3ex]

    \shortstack[l]{$FreeOverride_i \cdot$\\[-0.35ex]$Post_t$}
    & \HTEest{-0.0062}{0.0054}
    & \HTEest{-0.0145^{***}}{0.0039}
    & \HTEest{-0.0218^{***}}{0.0076}
    & \HTEest{-0.0269^{***}}{0.0087}
    & \HTEest{-0.0252^{***}}{0.0053}
    & \HTEest{-0.0434^{***}}{0.0140}\\[0.3ex]

    \shortstack[l]{$R^2$\\[-0.35ex]Observations}
    & \HTEstat{0.2820}{621{,}497}
    & \HTEstat{0.4088}{621{,}497}
    & \HTEstat{0.4584}{621{,}497}
    & \HTEstat{0.1668}{110{,}507}
    & \HTEstat{0.1106}{110{,}507}
    & \HTEstat{0.1737}{110{,}507}\\

    \midrule

    \textbf{Panel C}
    & \multicolumn{3}{c}{\textbf{Senior Worker}}
    & \multicolumn{3}{c}{\textbf{Junior Worker}}\\
    \cmidrule(lr){2-4}\cmidrule(l){5-7}

    \shortstack[l]{$ConstrainedOverride_i \cdot$\\[-0.35ex]$Post_t$}
    & \HTEest{0.0002}{0.0049}
    & \HTEest{-0.0014}{0.0033}
    & \HTEest{-0.0159^{***}}{0.0061}
    & \HTEest{0.0022}{0.0058}
    & \HTEest{-0.0038}{0.0046}
    & \HTEest{-0.0087}{0.0089}\\[0.3ex]

    \shortstack[l]{$FreeOverride_i \cdot$\\[-0.35ex]$Post_t$}
    & \HTEest{-0.0138^{**}}{0.0055}
    & \HTEest{-0.0156^{***}}{0.0041}
    & \HTEest{-0.0258^{***}}{0.0073}
    & \HTEest{-0.0087}{0.0065}
    & \HTEest{-0.0115^{**}}{0.0055}
    & \HTEest{-0.0083}{0.0132}\\[0.3ex]

    \shortstack[l]{$R^2$\\[-0.35ex]Observations}
    & \HTEstat{0.2004}{439{,}189}
    & \HTEstat{0.2189}{439{,}189}
    & \HTEstat{0.1781}{439{,}189}
    & \HTEstat{0.1458}{199{,}923}
    & \HTEstat{0.2139}{199{,}923}
    & \HTEstat{0.1848}{199{,}923}\\

    \midrule

    \edit{\textbf{Panel D}}
    & \multicolumn{3}{c}{\edit{\textbf{Top 50\% Machine--SKU Familiarity}}}
    & \multicolumn{3}{c}{\edit{\textbf{Low 50\% Machine--SKU Familiarity}}}\\
    \cmidrule(lr){2-4}\cmidrule(l){5-7}

    \edit{\shortstack[l]{$ConstrainedOverride_i \cdot$\\[-0.35ex]$Post_t$}}
    & \edit{\HTEest{-0.0005}{0.0038}}
    & \edit{\HTEest{-0.0028}{0.0026}}
    & \edit{\HTEest{-0.0207^{***}}{0.0065}}
    & \edit{\HTEest{0.0002}{0.0050}}
    & \edit{\HTEest{0.0006}{0.0042}}
    & \edit{\HTEest{0.0108^{*}}{0.0058}}\\[0.3ex]

    \edit{\shortstack[l]{$FreeOverride_i \cdot$\\[-0.35ex]$Post_t$}}
    & \edit{\HTEest{-0.0119^{***}}{0.0044}}
    & \edit{\HTEest{-0.0131^{***}}{0.0032}}
    & \edit{\HTEest{-0.0300^{***}}{0.0083}}
    & \edit{\HTEest{-0.0048}{0.0067}}
    & \edit{\HTEest{-0.0093^{*}}{0.0055}}
    & \edit{\HTEest{0.0085}{0.0076}}\\[0.3ex]

    \edit{\shortstack[l]{$R^2$\\[-0.35ex]Observations}}
    & \edit{\HTEstat{0.1722}{599{,}809}}
    & \edit{\HTEstat{0.1591}{599{,}809}}
    & \edit{\HTEstat{0.2324}{599{,}809}}
    & \edit{\HTEstat{0.0986}{461{,}456}}
    & \edit{\HTEstat{0.1498}{461{,}456}}
    & \edit{\HTEstat{0.1655}{461{,}456}}\\

    \bottomrule[1pt]

    \multicolumn{7}{@{}p{\linewidth}@{}}{
    {\tiny \textit{Notes.}
        Top-Sales SKUs are classified by the company as top sellers based on monthly sales, whereas Lower-Sales SKUs comprise all remaining SKUs. Growth-Stage SKUs have been in the product life cycle for at most $150$ days since their first fill, whereas Mature-Stage SKUs have been in the product life cycle for more than $150$ days. Senior workers have more than one year of tenure at the start of the experiment, whereas Junior workers have at most one year. Machine--SKU familiarity is measured by the worker's monthly number of pre-treatment replenishments of the same SKU at the same machine; the Top 50\% and Low 50\% groups are defined by a median split.} Robust standard errors clustered by worker are in parentheses. $^{*}p<0.1$; $^{**}p<0.05$; $^{***}p<0.01$.}
    \end{tabular}
    \endgroup
\end{table}

These patterns suggest considering product characteristics and workers' local experience when setting the cap. The tests assess policy effects within each subgroup. Free override's negative sales-probability estimates highlight the risk of unrestricted reductions, while profit consequences also depend on inventory savings and relative margin. These findings relate to research on private information in operations \citep{van2010ordering, brau2023demand, stip2024reluctant}. Complementing approaches that incorporate human information into algorithmic decisions \citep{sun2022predicting}, a cap lets workers select interventions without requiring the firm to identify each useful case in advance. The patterns also align with prior findings for experienced workers \citep{kwon2022human} and growth-stage SKUs \citep{kesavan2020field}.

\subsection{How Many Overrides Should Firms Allow?}\label{subsection:different-k}
The experiment evaluates a common cap of $K=2$. The number of useful override opportunities may differ across workers and assignments, so one cap need not suit every setting. We use simulations to examine how the sales--inventory trade-off changes when firms allow more or fewer overrides. We use the free-override group as the basis for these counterfactuals because it reveals the downward overrides workers chose without a cap. The constrained group helps approximate which of these opportunities workers would prioritize when discretion is limited. Specifically, we pool the observed downward overrides from both groups and use a logistic regression based on recorded replenishment conditions and workers' actual replenishment quantities to estimate the probability that an override comes from the constrained group. Within each free-group machine replenishment, we rank overrides by this probability and retain at most the top $K$. For retained overrides, we preserve the worker's chosen replenishment quantity; for the remaining downward overrides, we restore the algorithm's recommendation. All other replenishment quantities remain unchanged. 

We simulate outcomes day by day using the replenishment quantities under each cap, observed inventory immediately before replenishment, and the company's predicted demand for each day of the subsequent replenishment window. Holding the daily demand predictions and the observed replenishment window fixed across caps, we construct SKU-level \textit{Sales}, \textit{Sales Probability}, and \textit{Inventory} using the same definitions and replenishment window as in the main analysis, then average these outcomes within each machine replenishment. We calculate \textit{Profit} using the firm's sales--inventory objective and report differences relative to $K=0$, which restores the algorithm's recommendation for every downward override. The simulation evaluates cap performance under the company's daily demand forecasts but does not model uncertainty around those forecasts. This may understate the value of inventory held to protect against unexpectedly high demand and thereby affect the preferred override cap. The simulated outcomes and profit-maximizing cap are therefore conditional on these forecast paths and our approximation of workers' prioritization of override opportunities. \Cref{appendix:better override} provides the full procedure.

\begin{figure}[htbp]
    \FIGURE
    {\includegraphics[width=0.7\linewidth]{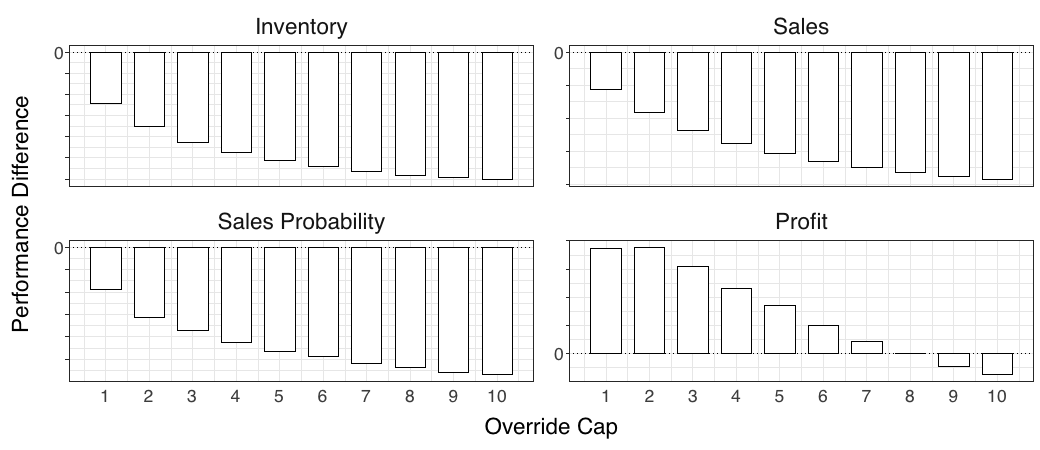}}
    {
        Simulated Performance under Different Override Caps \label{fig:performance}
    }
    {
        All outcomes are simulated relative to the no-downward-override benchmark ($K=0$). Nonzero numerical labels on the vertical axes are omitted to protect confidential company information. Profit is evaluated at the company's actual relative margin. To protect proprietary information, we do not disclose its exact value, but report that $r\in[0.8,0.95]$.
    }
\end{figure}

\Cref{fig:performance} shows that additional downward overrides reduce inventory as $K$ increases, but incremental savings diminish while losses in sales and sales probability accumulate. At the company's relative margin, simulated profit peaks at $K=2$ and then declines. This pattern is consistent with selective filtering: the highest-ranked overrides capture most of the value, whereas lower-ranked overrides increasingly impose sales costs. The aggregate analysis in \Cref{fig:service-level} of \Cref{appendix:better override} further shows that the profit-maximizing override cap $K^*(r)$ weakly decreases as $r$ rises because placing greater weight on sales makes additional downward overrides more costly. \Cref{fig:service-level-heterogeneity} further suggests that worker knowledge shifts this trade-off. Near the company's relative margin, the simulation selects a higher $K$ for senior than junior workers and for high- than low-familiarity machine--SKU pairs. Thus, additional override opportunities remain valuable over a wide range when workers are more likely to possess useful task-specific information, consistent with the heterogeneous effects in \Cref{subsection: HTE}.

\begin{figure}[htbp]
    \vspace{-5pt}
    \FIGURE
    {
        \subfigure[\textit{Worker Experience}]{
            \includegraphics[width=0.475\textwidth]
            {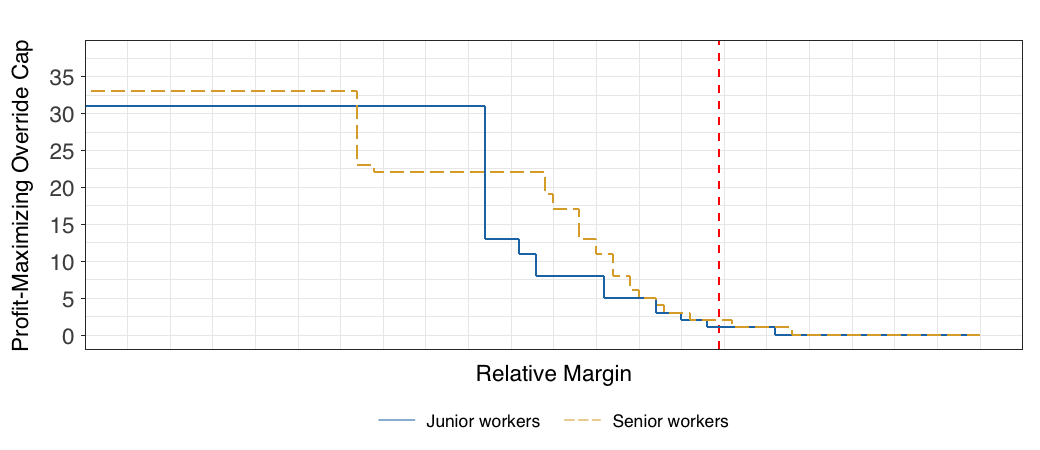}
            \label{fig:service-level-seniority}
        }
        \hfill
        \subfigure[\textit{Machine--SKU Familiarity}]{
            \includegraphics[width=0.475\textwidth]
            {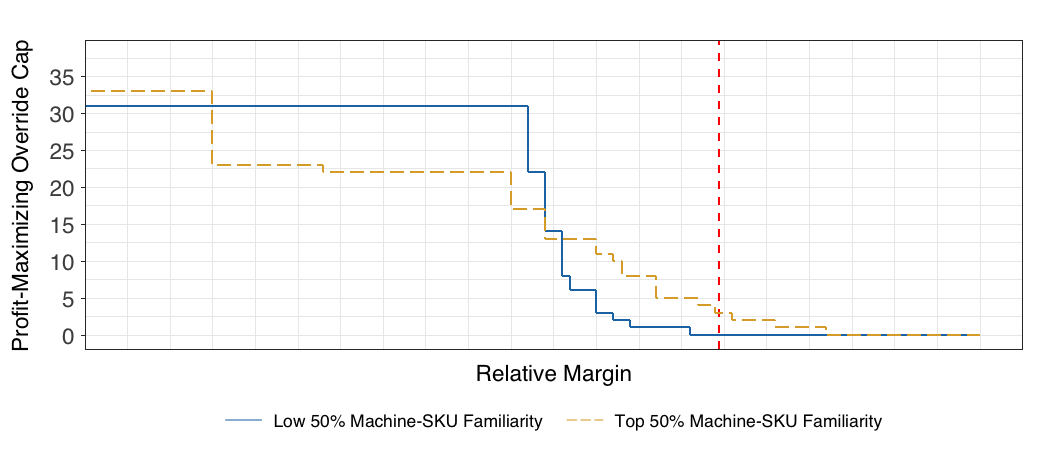}
            \label{fig:service-level-familiarity}
        }
    }
    {Simulated Profit-Maximizing Override Caps by Worker Experience and Machine--SKU Familiarity
    \label{fig:service-level-heterogeneity}}
    {A larger $K$ permits more downward overrides. The red dashed vertical line marks the company's relative margin $r$. Its exact value is proprietary but lies in $[0.8,0.95]$.}
\end{figure}

Overall, the results suggest that firms should choose the cap to retain useful inventory reductions while limiting the sales costs of additional overrides. This balance depends on the firm's relative margin and workers' local information. The simulated choice of $K=2$ reflects Fengyi Technology's economics, workforce, and operations; following the algorithm ($K=0$) may be appropriate where local information is weak. Other firms can start below their typical override count under free discretion and adjust the cap to their relative margin, worker experience, and machine--SKU familiarity. Differentiating caps also requires specifying how SKUs share each machine's override budget. 

\section{Conclusion}\label{section: conclusion}
We propose a simple, scalable solution to improving human supervision of autonomous AI: a constrained override policy. By imposing a cardinality constraint on discretionary overrides, the policy encourages workers to prioritize the opportunities they perceive as most valuable, which can improve override quality. Unlike approaches that require identifying bias or private information in advance, our method is straightforward to implement. In a large-scale field experiment on Fengyi Technology’s smart-vending network, constrained override reduced inventory with no statistically significant sales decline, in contrast to the performance losses observed under the free override policy. Following the experiment, Fengyi Technology scaled the policy to more than 180,000 machines and 5,000 SKUs across over 70 metropolitan areas, contributing to an estimated \$1,000,000 annual reduction in inventory costs.

\ACKNOWLEDGMENT{The authors sincerely thank Fengyi Technology for their support. In particular, we are grateful to Xinning Shan, CEO of Fengyi Technology, for his generous support. We also extend our gratitude to Xiaoxia Chen for her invaluable assistance in setting up the randomized field experiments.}

\theendnotes

\bibliographystyle{informs2014}

\let\oldbibliography\thebibliography
 \renewcommand{\thebibliography}[1]{%
    \oldbibliography{#1}
    \setlength{\itemsep}{3pt}
}

\bibliography{reference}

\newpage

\begin{APPENDICES}
\pagenumbering{arabic}
\setcounter{page}{1}

\renewcommand{\theHsection}{A\arabic{section}}

\section{Additional Information and Results from Section~\ref{section: experiment}}\label{appendix:sec3}

\subsection{Empirical Context and Experiment Design}\label{appendix: exp}

\begin{figure}[htbp]
    \vspace{-10pt}
    
    \FIGURE
    {%
        \subfigure[\textit{Office Areas}]{%
            \includegraphics[width=0.4\textwidth]{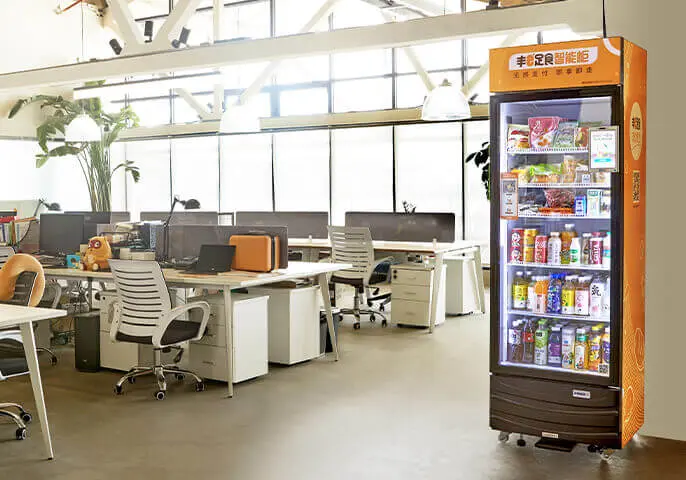}%
        }%
        \hfill
        \subfigure[\textit{Hotel Lobbies}]{%
            \includegraphics[width=0.4\textwidth]{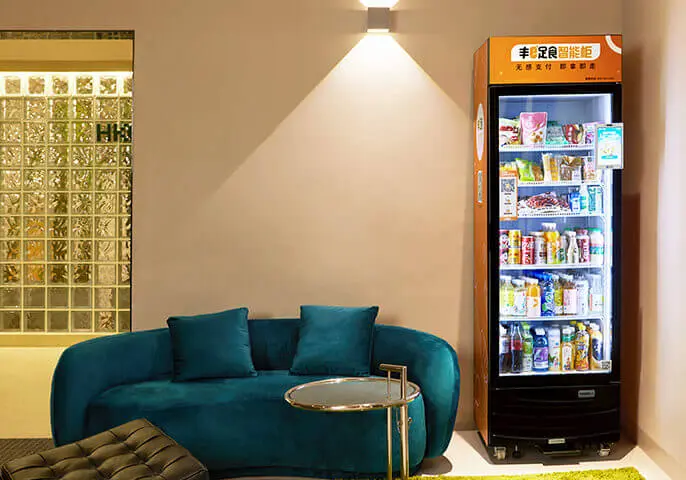}%
        }%
    }
    {Smart Vending Machines by Fengyi Technology
     \label{fig: smart vending machine}}
    {The figure shows photos of smart vending machines deployed by Fengyi Technology in diverse offline environments.}
    
    \vspace{-15pt}
\end{figure}

\begin{figure}[!htbp]
    \vspace{-10pt}
    \FIGURE
    {\includegraphics[width=0.6\linewidth]{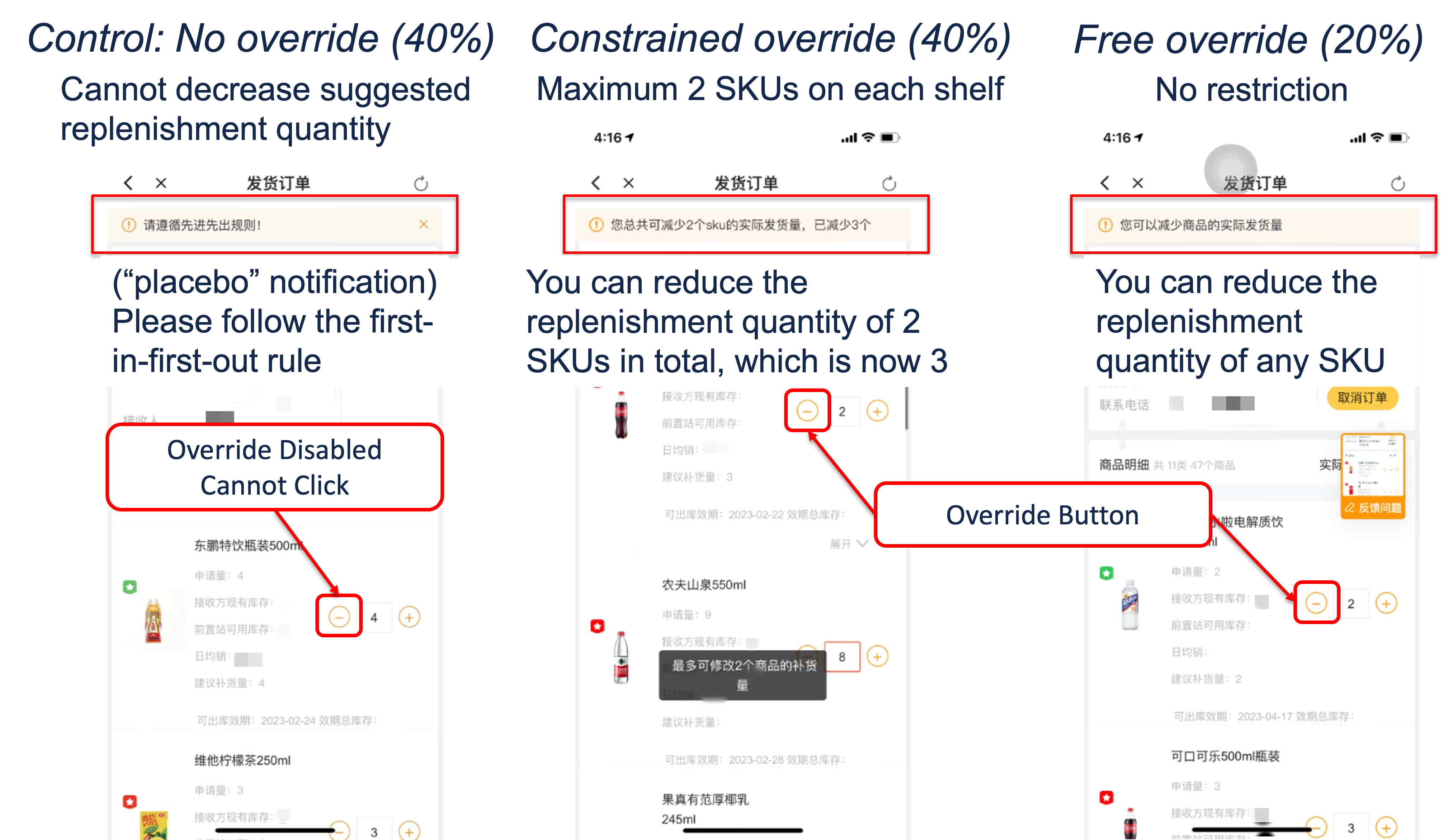}}
    {Screenshots of the App for Different Groups \label{fig: screenshots for different groups}}
    {Workers receive algorithmic suggestions via the app, along with inventory and sales data. Only downward overrides are restricted; upward overrides remain unrestricted.}
    \vspace{-15pt}
\end{figure}

\begin{table}[!htbp] 
    \linespread{0.5}
    \centering
    \caption{Randomization Check}
    \label{table: randomization check}
    \scriptsize
    \begin{tabular}{@{\extracolsep{5pt}}lcccccc}
        \toprule[1pt]
            & \multicolumn{3}{c}{Means} & \multicolumn{3}{c}{Differences} \\[0.5ex]
            \cline{2-4} \cline{5-7}\\[-0.5ex]
            Feature & No & Constrained & Free & Constrained & Free & Constrained \\
            &  &  &  & Override & Override & Override \\
            & Override & Override & Override & $-$ No & $-$ No & $-$ Free \\
            & & & & Override & Override &Override\\
        \midrule
        
        \textit{Tenure} & $0.710$ & $0.664$ & $0.651$ & $-0.047$ & $-0.059$ & $0.012$\\
        (\textit{Binary: Senior/Junior}) & $[0.455]$ & $[0.474]$ & $[0.479]$ & $(0.044)$ & $(0.055)$ & $(0.056)$\\ [0.5ex]
        
        \textit{Managed Machines} & $96.757$ & $94.916$ & $94.978$ & $-1.841$ & $-1.779$ & $-0.062$\\
        (\textit{Unit: Machine}) & $[28.222]$ & $[23.971]$ & $[23.922]$ & $(2.486)$ & $(2.976)$ & $(2.798)$\\[0.5ex]
        
        \textit{Top-Sales Machines} & $0.137$ & $0.145$ & $0.142$ & $0.008$ & $0.005$ & $0.003$\\
        (\textit{Share}) & $[0.083]$ & $[0.082]$ & $[0.081]$ & $(0.008)$ & $(0.010)$ & $(0.009)$\\[0.5ex]
        
        \textit{Top-Sales SKUs} & $0.174$ & $0.178$ & $0.178$ & $0.004$ & $0.005$ & $-0.0003$\\
        (\textit{Share}) & $[0.061]$ & $[0.059]$ & $[0.064]$ & $(0.006)$ & $(0.007)$ & $(0.007)$\\[0.5ex]
        
        \textit{Growth-Stage SKUs} & $0.554$ & $0.551$ & $0.564$ & $-0.003$ & $0.010$ & $-0.013$\\
        (\textit{Share}) & $[0.085]$ & $[0.080]$ & $[0.076]$ & $(0.008)$ & $(0.009)$ & $(0.009)$\\[0.5ex]
        
        \textit{Average Sales per Machine} & \multicolumn{3}{c}{} & $0.0068$ & $0.0045$ & $0.0023$ \\
        (\textit{Unit: Item}) &  &  &  & $(0.0080)$ & $(0.0097)$ & $(0.0104)$\\[0.5ex] 
        
        \textit{Average Sales Probability per Machine} & \multicolumn{3}{c}{Data anonymized} & $0.0023$ & $0.0016$ & $0.0007$\\
        (\textit{Probability}) &  &  &  & $(0.0030)$ & $(0.0037)$ & $(0.0037)$\\ [0.5ex]
        
        \textit{Average Inventory per Machine} & \multicolumn{3}{c}{} & $-0.0230$ & $0.0056$ & $-0.0286$\\
        (\textit{Unit: Item}) &  &  &  & $(0.0435)$ & $(0.0541)$ & $(0.0557)$\\ [0.5ex]
        
        \textit{Intervened Machines}  & $2.796$ & $3.108$ & $3.066$ & $0.312$ & $0.270$ & $0.042$ \\
        (\textit{Unit: Machine}) & $[2.952]$ & $[3.038]$ & $[2.963]$ & $(0.284)$ & $(0.346)$ & $(0.349)$\\ [0.5ex]
        
        \textit{\# of Downward SKUs per Machine}  & $1.502$ & $1.710$ & $1.850$ & $0.208$ & $0.348$ & $-0.140$ \\
        (\textit{Unit: SKU}) & $[1.735]$ & $[2.097]$ & $[2.021]$ & $(0.183)$ & $(0.226)$ & $(0.239)$\\ [0.5ex]
        
        $n$ & $221$ & $223$ & $109$ &  &  & \\
        
        \bottomrule \\[-0.8ex]
        
        \multicolumn{7}{l}{\textit{Note.} Standard deviations of each group are in brackets, and standard errors are in parentheses.}\\[0.2ex]
        \multicolumn{7}{l}{Intervened Machines counts machines with downward overrides per worker-day.}\\[0.2ex]
        \multicolumn{7}{l}{Statistics on downward-overridden SKUs per machine exclude machines with no downward overrides.}\\[0.2ex]
        \multicolumn{7}{l}{$^{*}p<0.1$; $^{**}p<0.05$; $^{***}p<0.01$}\\
    \end{tabular}
\end{table}

\subsection{Proofs of the Bayesian Model}\label{appendix: bayesian}
\proof{Proof of \Cref{lemma:conditional_positive}.}
With $\eta_k\sim\mathcal N(\mu,\sigma_\eta^2)$, $s_k=\eta_k+\epsilon_k$, and $\epsilon_k\sim\mathcal N(0,\sigma^2)$, normal conjugacy gives $\mathbb E[\eta_k\mid s_k=s]=\mu+(\sigma_\eta^2/\tau^2)(s-\mu)$, where $\tau^2=\sigma_\eta^2+\sigma^2$. Write $s_k=\mu+\tau U_k$ with $U_k\sim\mathcal N(0,1)$. By exchangeability, conditional on selection among the top $K$ signals, each of the first $K$ ranks is equally likely. Hence $\mathbb E[s_k\mid s_k\ge s_{(K)}]=\mu+\tau a(K)$, where $a(K)=K^{-1}\sum_{j=1}^{K}\mathbb E[U_{(j)}]$. Applying iterated expectation gives $\mathbb E[\eta_k\mid k\in A^{\mathrm{Constrained}}]=\mu+(\sigma_\eta^2/\tau)a(K)$, which is positive if $a(K)>-\mu\tau/\sigma_\eta^2$. \Halmos

\proof{Proof of \Cref{prop:filtering}.}
No override gives $A^{\mathrm{No}}=\emptyset$ and zero incremental profit. Under free override, a SKU is selected when $s_k>0$, with probability $\Phi(\mu/\tau)$. Standardizing $s_k=\mu+\tau U_k$ gives $\mathbb E[s_k\mid s_k>0]=\mu+\tau h(-\mu/\tau)$, where $h(x)=\phi(x)/[1-\Phi(x)]$ is the inverse Mills ratio. Substituting into the posterior mean from \Cref{lemma:conditional_positive} yields $\mathbb E[\eta_k\mid s_k>0]=\mu+(\sigma_\eta^2/\tau)h(-\mu/\tau)$. Thus, expected total incremental profit is $\mathbb E[\sum_{k\in A^{\mathrm{Free}}}\eta_k]=N\Phi(\mu/\tau)\{\mu+(\sigma_\eta^2/\tau)h(-\mu/\tau)\}$, which is negative if $h(-\mu/\tau)<-\mu\tau/\sigma_\eta^2$. Under constrained override, selecting the $K$ highest signals gives $\mathbb E[\sum_{k\in A^{\mathrm{Constrained}}}\eta_k]=K\{\mu+(\sigma_\eta^2/\tau)a(K)\}$ by \Cref{lemma:conditional_positive}. This is positive if $a(K)>-\mu\tau/\sigma_\eta^2$. The stated ordering therefore holds when $h(-\mu/\tau)<-\mu\tau/\sigma_\eta^2<a(K)$. \Halmos

\begin{remark}
    Under the baseline assumptions, the sufficient conditions in \Cref{prop:filtering} also establish the policy ordering for $A^{\mathrm{Constrained}} = \{k:s_k>0,\ s_k\geq s_{(K)}\}$ without approximation, since excluding non-positive assessments cannot reduce expected total incremental profit.
\end{remark}

\begin{proposition}[Model-Implied Bound on the Override Cap] \label{prop:optimal_K_ratio}
    \edit{Under the conditions of \Cref{prop:filtering} and the top-$K$ approximation used in the proof, the profit-maximizing override cap $K^*$ satisfies $K^* < \mathbb{E}[K_{\mathrm{free}}] = N\Phi\left(\frac{\mu}{\tau}\right)$.}
\end{proposition}

\proof{Proof of \Cref{prop:optimal_K_ratio}.}
By \Cref{lemma:conditional_positive}, expected constrained-override profit is $V(K)=K\{\mu+(\sigma_\eta^2/\tau)a(K)\}$. Let $\bar K=\lceil\mathbb E[K_{\mathrm{free}}]\rceil=\lceil N\Phi(\mu/\tau)\rceil$. For $K\ge\bar K$,
$a(K)\le a(\bar K)\le h\{\Phi^{-1}(1-\bar K/N)\}\le h(-\mu/\tau)$.
The first inequality follows because $a(K)$ is non-increasing; the second follows from \citet[Theorem 1]{rychlik19986}, applied to the average of the largest $\bar K$ standard-normal order statistics. The third uses $\bar K/N\ge\Phi(\mu/\tau)$ and the monotonicity of $\Phi^{-1}$ and $h$ \citep[Corollary 2]{bagnoli2005log}. Under \Cref{prop:filtering}, $\mu+(\sigma_\eta^2/\tau)h(-\mu/\tau)<0$, so $V(K)<0$ for $K\ge\bar K$, whereas some smaller cap gives positive profit. Thus $K^*<\bar K$; since $K^*$ is integer, $K^*<\mathbb E[K_{\mathrm{free}}]$. \Halmos

\edit{\Cref{prop:optimal_K_ratio} suggests setting the cap below the expected number of overrides under free discretion. Choosing a numerical cap also requires evidence on the firm's replenishment operations, as discussed in \Cref{subsection: experiment design,section:betteroverride}.}

\subsection{A Rational Inattention Model}\label{appendix: alternative models}
The baseline model holds assessment precision fixed. Here, workers choose how much information to acquire before selecting overrides, using the same precision for all $N$ SKUs \citep{simon1955behavioral,sims2003implications}. More precise assessments are costly, so the worker weighs their expected value against the cost of acquiring information.

As in the baseline, true override values are independent across SKUs, with $\eta_k\sim\mathcal N(\mu,\sigma_\eta^2)$, $\mu<0$, and $\sigma_\eta^2>0$, while workers form beliefs using the optimistic prior $\eta_k\sim\mathcal N(0,\sigma_\eta^2)$. \edit{The worker chooses a common precision $\lambda\in\Lambda=\{0\}\cup\Lambda_+$, where $\Lambda_+\subseteq(0,\infty)$.} For $\lambda>0$, the assessment is $s_k=\eta_k+\epsilon_k$, with independent errors $\epsilon_k\sim\mathcal N(0,1/\lambda)$ that are also independent of override values. Write $\tau(\lambda)=\sqrt{\sigma_\eta^2+1/\lambda}$. The information cost $C(\lambda)=(\kappa N/2)\log(1+\lambda\sigma_\eta^2)$, with $\kappa>0$, is proportional to the total mutual information acquired about the $N$ SKUs. \edit{The option $\lambda=0$ means acquiring no information and has zero cost.}

The worker chooses precision before observing assessments, anticipating the override sets $\mathcal A$ permitted by the policy. Using subjective expectation $\mathbb E^W$, the worker solves $\max_{\lambda\in\Lambda} \left\{ \mathbb E^W\!\left[ \max_{A\in\mathcal A}\sum_{k\in A}\mathbb E^W[\eta_k\mid\boldsymbol{s},\lambda] \right]-C(\lambda) \right\}$. At $\lambda=0$, posterior expected values remain zero under the worker's prior; we assume that she follows the algorithm when indifferent. Under no override, $A^{\mathrm{No}}=\emptyset$ and the worker chooses $\lambda^{\mathrm{No}}=0$. For $\lambda>0$, the subjective posterior mean is $\mathbb E^W[\eta_k\mid\boldsymbol{s},\lambda]=(\sigma_\eta^2/\tau(\lambda)^2)s_k$. It is increasing in $s_k$ and has the same sign, so the rules remain $A^{\mathrm{Free}}(\lambda)=\{k:s_k>0\}$ and $A^{\mathrm{Constrained}}(\lambda)=\{k:s_k>0,\ s_k\ge s_{(K)}\}$. We compare the expected incremental profit from implemented overrides before information costs.

\begin{proposition}[Selective Filtering under Rational Inattention]\label{prop:RI_filtering}
    Suppose $\lambda^{\mathrm{Constrained}}>0$ and $\lambda^{\mathrm{Free}}>0$ are the optimal precisions under the two active policies, and write $\tau_C=\tau(\lambda^{\mathrm{Constrained}})$ and $\tau_F=\tau(\lambda^{\mathrm{Free}})$. The ordering
    \begin{equation*}
        \mathbb E\!\left[\sum_{k\in A^{\mathrm{Constrained}}(\lambda^{\mathrm{Constrained}})}\eta_k\right] > \mathbb E\!\left[\sum_{k\in A^{\mathrm{No}}}\eta_k\right] = 0 > \mathbb E\!\left[\sum_{k\in A^{\mathrm{Free}}(\lambda^{\mathrm{Free}})}\eta_k\right]
    \end{equation*}
    holds if $a(K)>-\mu\tau_C/\sigma_\eta^2$ and $h(-\mu/\tau_F)<-\mu\tau_F/\sigma_\eta^2$, where $a(K)$ and $h$ are defined in the baseline model.
\end{proposition}

\proof{Proof of \Cref{prop:RI_filtering}.}
At the chosen constrained-policy precision, selecting all top $K$ assessments gives expected total override value $K\{\mu+(\sigma_\eta^2/\tau_C)a(K)\}$. For any nonpositive assessment, the true posterior mean is $\mu(1-\sigma_\eta^2/\tau_C^2)+(\sigma_\eta^2/\tau_C^2)s_k<0$, even conditional on all other assessments. The actual rule excludes these overrides, so its expected total value is at least the top-$K$ expression and is positive under the first condition. At the chosen free-policy precision, expected total override value is $N\Phi(\mu/\tau_F)\{\mu+(\sigma_\eta^2/\tau_F)h(-\mu/\tau_F)\}$, which is negative under the second condition. No override has zero incremental value. These expectations use the true distribution, whereas precision is chosen using the worker's subjective objective. \Halmos

\Cref{prop:RI_filtering} extends selective filtering to costly information acquisition: under the stated conditions, constrained override improves expected incremental profit before information costs even when the policies induce different assessment precision.

\subsection{Cross-SKU Correlation} \label{appendix: correlated signals}
The baseline model assumes that override values and assessment errors are independent across SKUs. To capture cross-SKU correlation, we use a common-component specification, analogous to \citet{glasserman2005importance}. Specifically, let $\eta_k = \mu + \sigma_\eta \left( \sqrt{\rho_\eta} Z_\eta + \sqrt{1 - \rho_\eta} Z_{\eta k} \right)$ and $\epsilon_k = \sigma \left( \sqrt{ \rho_\epsilon} Z_\epsilon + \sqrt{1-\rho_\epsilon} Z_{\epsilon k} \right)$, where all $Z$ terms are mutually independent standard normal variables and $\rho_\eta, \rho_\epsilon \in [0,1)$. The parameter $\rho_\eta$ captures common variation in override values, such as shared demand conditions, whereas $\rho_\epsilon$ captures common assessment error arising from factors such as worker confidence or fatigue. The marginal distributions of $\eta_k$ and $\epsilon_k$ remain the same as in the baseline model.

Common components cancel when comparing SKUs $k$ and $l$: $\eta_k-\eta_l = \sigma_\eta\sqrt{1-\rho_\eta} (Z_{\eta k}-Z_{\eta l})$ and $s_k-s_l = (\eta_k-\eta_l) + \sigma\sqrt{1-\rho_\epsilon} (Z_{\epsilon k}-Z_{\epsilon l})$. Let $\tau_I^2 \coloneqq \sigma_\eta^2(1-\rho_\eta) + \sigma^2(1-\rho_\epsilon)$. Therefore, $\mathbb E^W[\eta_k-\eta_l \mid\boldsymbol{s}]$ depends on $\boldsymbol{s} = (s_1,\ldots,s_N)$ only through the corresponding signal difference $s_k-s_l$. Under the worker's zero-mean prior,
\begin{equation*}
    \mathbb E^W[\eta_k\mid\boldsymbol{s}] - \mathbb E^W[\eta_l\mid\boldsymbol{s}] = \mathbb E^W[\eta_k-\eta_l\mid s_k-s_l] = \frac{\sigma_\eta^2(1-\rho_\eta)} {\tau_I^2} (s_k-s_l).
\end{equation*}
Thus, a higher assessment implies a higher perceived incremental profit. When the cap binds, the worker selects the $K$ SKUs with the highest assessments.

Under the baseline's top-$K$ approximation, the calculation in \Cref{lemma:conditional_positive} gives $\mathbb E\left[ \eta_k \,\middle|\, k\in A^{\mathrm{Constrained}} \right] = \mu + \sigma_\eta^2(1-\rho_\eta)a(K)/\tau_I$. Therefore, constrained overrides are profitable on average if $a(K) > -\mu\tau_I/\{\sigma_\eta^2(1-\rho_\eta)\}$. Greater correlation in override values weakens selective filtering by reducing variation across SKUs. By contrast, greater correlation in assessment errors leaves less SKU-specific noise to distort assessment ordering, allowing workers to identify better override opportunities more accurately. 

For the policy comparison, suppose $\rho_\eta=\rho_\epsilon=\rho$. The common correlation structure gives $\mathbb E^W[\eta_k\mid\boldsymbol{s}]=\{\sigma_\eta^2/(\sigma_\eta^2+\sigma^2)\}s_k$, so free override still selects $s_k>0$. Its expected incremental profit is unchanged because $(\eta_k,s_k)$ has the same marginal joint distribution. With $\tau_I=\sqrt{1-\rho}\tau$ and $\tau^2\coloneqq\sigma_\eta^2+\sigma^2$, the policy ordering in \Cref{prop:filtering} holds if $h(-\mu/\tau)<-\mu\tau/\sigma_\eta^2<\sqrt{1-\rho} \ a(K)$. As $\rho$ approaches one, differences in override values across SKUs vanish, eliminating the benefit of prioritizing the highest assessments. Constrained override may still reduce the losses from free override, but it need not outperform no override. Thus, strict cross-SKU independence is not necessary for selective filtering; the mechanism still operates when workers' assessments distinguish relatively better override opportunities across SKUs.

Operationally, the common-component extension shows that shared influences need not undermine selective filtering when workers retain sufficiently informative SKU-specific assessments; the selection gain weakens as common variation in override values leaves less scope to distinguish better opportunities. The evidence linking machine--SKU familiarity to override safety in \Cref{table:free_override_selective_filtering} and \Cref{subsection:information_filter} supports the practical relevance of such information, while the randomized policy comparison evaluates constrained override without imposing cross-SKU independence.

\section{Additional Analyses for Section~\ref{section: results}}\label{appendix: main results}
Figure~\ref{fig: model free} plots the daily average machine-level sales (top panel) and inventory levels (bottom panel). Each point represents the average value across all treated machines operated by active workers in that group on a given day. The data span April 1 to June 18, 2023, with the vertical line indicating May 4, 2023, the start of the experimental override policies. Prior to treatment, the three groups exhibit parallel trends, consistent with the random assignment of override policies at the worker level. After the intervention, meaningful differences emerge. The constrained override group achieves the highest average sales while maintaining moderate inventory levels. The free override group holds the lowest inventory but suffers from lower sales, suggesting excessive reductions. 

\begin{figure}[htbp]
    \FIGURE
    {\includegraphics[width=0.8\textwidth]{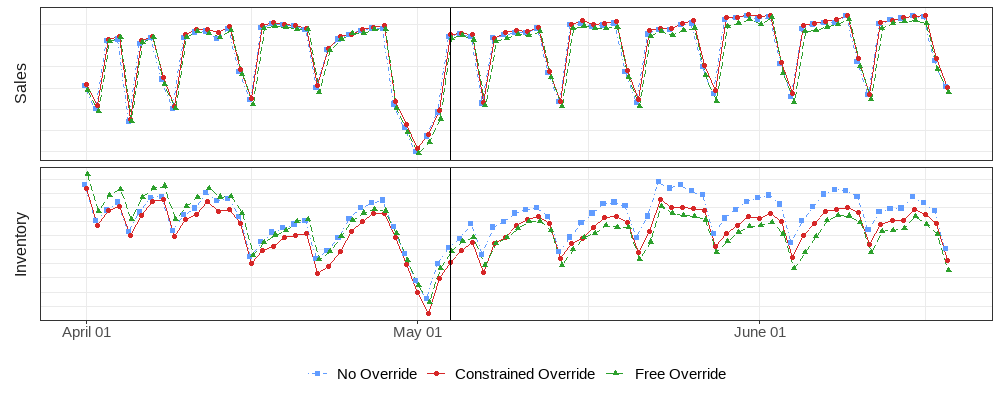}}
    {Model-Free Evidence ($y$-axis Anonymized) \label{fig: model free}}
    {}
\end{figure}

\subsection{Robustness Checks on Main Policy Effect} \label{appendix: main result}
To check whether the results depend on using log outcomes, columns (1)--(2) of \Cref{table:main_results_robustness} measure sales and inventory in items. Constrained override reduces inventory by $0.0860$ items ($p<0.05$), with a small positive but statistically insignificant sales estimate. Free override reduces inventory by $0.1452$ items ($p<0.01$) and sales by $0.0327$ items ($p<0.05$). The sales--inventory pattern therefore remains similar when outcomes are measured in their original units.
\begin{table}[!htbp]
    \linespread{0.5}
    \centering
    \scriptsize
    \caption{Robustness Checks for the Main Results}
    \label{table:main_results_robustness}

    \begin{tabular}{@{\extracolsep{4pt}}lccccc}
        \toprule

        & \multicolumn{2}{c}{Level Outcomes}
        & \multicolumn{3}{c}{Excluding Weekends and Public Holidays} \\

        \cmidrule(lr){2-3}
        \cmidrule(lr){4-6}

        & $\textit{Sales}$
        & $\textit{Inventory}$
        & $\log(\textit{Sales}+1)$
        & $\textit{Sales Probability}$
        & $\log(\textit{Inventory}+1)$ \\

        & (1) & (2) & (3) & (4) & (5) \\

        \midrule

        $Constrained\ Override_i \times Post\ Treatment_t$
        & $0.0084$
        & $-0.0860^{**}$
        & $0.0018$
        & $-0.0022$
        & $-0.0117^{**}$ \\

        &
        $(0.0169)$
        & $(0.0408)$
        & $(0.0039)$
        & $(0.0028)$
        & $(0.0050)$ \\[0.5ex]

        $Free\ Override_i \times Post\ Treatment_t$
        & $-0.0327^{**}$
        & $-0.1452^{***}$
        & $-0.0128^{***}$
        & $-0.0145^{***}$
        & $-0.0220^{***}$ \\

        &
        $(0.0162)$
        & $(0.0524)$
        & $(0.0042)$
        & $(0.0033)$
        & $(0.0065)$ \\

        \midrule

        Worker, (City--Date) FE
        & Yes & Yes & Yes & Yes & Yes \\

        R$^2$
        & $0.1739$
        & $0.1039$
        & $0.1651$
        & $0.1881$
        & $0.1681$ \\

        Observations
        & $639{,}112$
        & $639{,}112$
        & $532{,}592$
        & $532{,}592$
        & $532{,}592$ \\

        \bottomrule\\[-0.8ex]

        \multicolumn{4}{l}{\textit{Note.} Robust standard errors clustered by each worker are in parentheses.}\\
        \multicolumn{4}{l}{$^{*}p<0.1$; $^{**}p<0.05$; $^{***}p<0.01$} \\ 
    \end{tabular}
\end{table}
To assess whether weekends and public holidays drive the findings, columns (3)--(5) of \Cref{table:main_results_robustness} exclude these days. Constrained override continues to reduce inventory ($-0.0117$, $p<0.05$), with no statistically significant change in sales or sales probability. Free override has negative coefficients for sales ($-0.0128$), sales probability ($-0.0145$), and inventory ($-0.0220$), all significant at $p<0.01$. The same sales--inventory pattern therefore remains in the weekday sample.

\begin{table}[!htbp]
    \linespread{0.5}
    \centering
    \scriptsize
    \caption{Causal Impact of Free and Constrained Override Policies (Machine-Level Weekly DID Estimates)}
    \label{table:main_results_by_week}
    \begin{tabular}{@{\extracolsep{5pt}}lcccc}
        \toprule
            
        & $\log(\textit{Sales}+1)$
        & \textit{Sales Probability}
        & $\log(\textit{Inventory}+1)$
        & \textit{Profit}\\[0.5ex]

        \midrule
            
        \multicolumn{5}{l}{\textbf{Constrained override}}\\
        Pre~1 & $0.0008\ (0.0048)$ & $-0.0022\ (0.0029)$ & $-0.0037\ (0.0059)$ & $0.0052\ (0.0172)$\\
            
        Week~0 & $0.0028\ (0.0031)$ & $-0.0001\ (0.0029)$ & $-0.0155^{***}\ (0.0054)$ & $0.0128^{*}\ (0.0074)$\\
            
        Week~1 & $0.0016\ (0.0038)$ & $-0.0007\ (0.0029)$ & $-0.0045\ (0.0061)$ & $0.0071\ (0.0122)$\\
            
        Week~2 & $-0.0027\ (0.0041)$ & $-0.0036\ (0.0032)$ & $-0.0147^{**}\ (0.0064)$ & $-0.0022\ (0.0123)$\\
            
        Week~3 & $-0.0006\ (0.0047)$ & $-0.0030\ (0.0034)$ & $-0.0160^{**}\ (0.0065)$ & $-0.0063\ (0.0262)$\\
            
        Week~4 & $-0.0003\ (0.0062)$ & $-0.0067^{*}\ (0.0038)$ & $-0.0130^{*}\ (0.0068)$ & $0.0276\ (0.0259)$\\
            
        Week~5 & $0.0013\ (0.0063)$ & $-0.0032\ (0.0050)$ & $-0.0162^{**}\ (0.0076)$ & $0.0207\ (0.0223)$\\
            
        Week~6 & $0.0099\ (0.0097)$ & $0.0050\ (0.0076)$ & $-0.0128\ (0.0092)$ & $0.0520\ (0.0347)$\\
            
        \multicolumn{5}{l}{\textbf{Free override}}\\
        Pre~1 & $-0.0010\ (0.0049)$ & $-0.0036\ (0.0029)$ & $-0.0060\ (0.0064)$ & $0.0039\ (0.0171)$\\
            
        Week~0 & $-0.0055\ (0.0039)$ & $-0.0063^{*}\ (0.0034)$ & $-0.0154^{**}\ (0.0060)$ & $-0.0106\ (0.0103)$\\
            
        Week~1 & $-0.0039\ (0.0044)$ & $-0.0088^{***}\ (0.0033)$ & $-0.0152^{**}\ (0.0075)$ & $-0.0027\ (0.0131)$\\
            
        Week~2 & $-0.0115^{**}\ (0.0051)$ & $-0.0166^{***}\ (0.0039)$ & $-0.0109\ (0.0084)$ & $-0.0238\ (0.0157)$\\
            
        Week~3 & $-0.0152^{***}\ (0.0050)$ & $-0.0177^{***}\ (0.0040)$ & $-0.0136^{*}\ (0.0081)$ & $-0.0452^{*}\ (0.0253)$\\
            
        Week~4 & $-0.0171^{***}\ (0.0060)$ & $-0.0194^{***}\ (0.0047)$ & $-0.0260^{***}\ (0.0085)$ & $-0.0286\ (0.0195)$\\
            
        Week~5 & $-0.0189^{***}\ (0.0072)$ & $-0.0199^{***}\ (0.0060)$ & $-0.0343^{***}\ (0.0102)$ & $-0.0270\ (0.0215)$\\
            
        Week~6 & $-0.0114\ (0.0104)$ & $-0.0077\ (0.0088)$ & $-0.0333^{***}\ (0.0108)$ & $-0.0260\ (0.0282)$\\
            
        \midrule
            
        Worker, (City--Date) FE
        & Yes & Yes & Yes & Yes\\
            
        R$^{2}$
        & $0.1802$
        & $0.2121$
        & $0.1759$
        & $0.1739$\\
            
        Observations
        & $639{,}112$
        & $639{,}112$
        & $639{,}112$
        & $639{,}112$\\
            
        \bottomrule\\[-0.8ex]
            
        \multicolumn{5}{p{0.95\textwidth}}{
        \textit{Note.} Rows report policy-by-week interactions; Week~0 is the first treatment week. Pre~1 is the preceding week. Robust standard errors clustered by worker are in parentheses.}\\
            
        \multicolumn{5}{l}{
        $^{*}p<0.1$; $^{**}p<0.05$; $^{***}p<0.01$.}\\
    \end{tabular}
\end{table}

To assess when the policy effects emerge and whether they persist, we estimate the following weekly DID model, with results reported in \Cref{table:main_results_by_week}:
\begin{equation}\label{eq:did_machine_dynamic}
    \begin{aligned}
        \textit{Outcome}_{ijt} = \alpha &+ \beta_{1,-1} (\textit{ConstrainedOverride}_i \cdot \textit{PreWeek1}_t) + \sum_{w=0}^{6} \delta_{1w} (\textit{ConstrainedOverride}_i \cdot \textit{PostWeek}_{wt}) \\
        &+ \beta_{2,-1}  (\textit{FreeOverride}_i \cdot \textit{PreWeek1}_t) + \sum_{w=0}^{6} \delta_{2w} (\textit{FreeOverride}_i \cdot \textit{PostWeek}_{wt}) \\
        &+ \gamma_i + \delta_{c(j)t} + \varepsilon_{ijt}.
    \end{aligned}
\end{equation}
We estimate separate policy effects for the seven treatment weeks and the week immediately before treatment. The pre-treatment estimates are statistically insignificant. Constrained override has negative inventory estimates in every treatment week, while free override has negative sales and sales-probability estimates throughout. \Cref{subsection:algorithm_suggestions_learning} examines what this timing suggests about human learning.

\subsection{Changes in Replenishment and Overriding Behavior} \label{subsection: override analysis}
We examine how discretion changes both the number of SKUs workers reduce and the reduction per selected SKU. In \Cref{fig: Downward SKU}, workers under free override frequently reduce several SKUs, while those under constrained override generally stay within the cap of $K=2$. All groups retain exceptions for operational needs, such as warehouse stockouts, which permit limited reductions under no override and additional reductions beyond the cap.

\begin{figure}[htbp]
    \FIGURE
    {\includegraphics[width=0.7\linewidth]{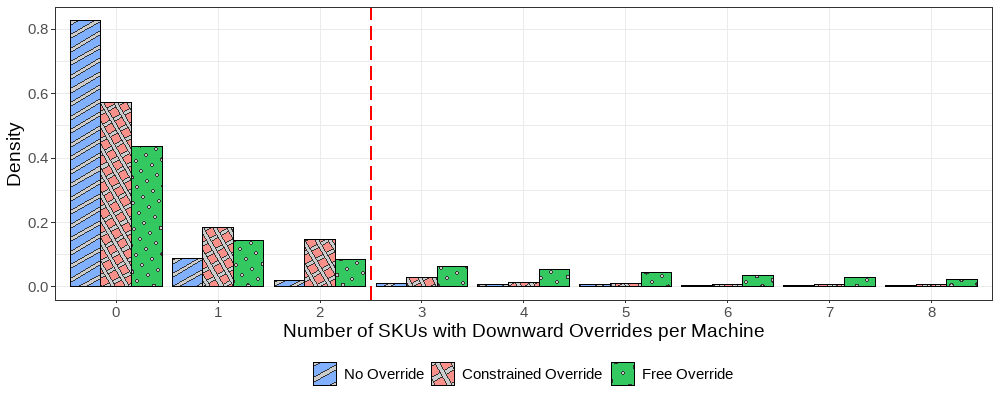}}
    {Number of Downward SKU Overrides per Machine by Policies \label{fig: Downward SKU}}
    {The histogram includes machines with no (0) downward overrides.}
\end{figure}

\Cref{table:override_analysis} summarizes downward and upward overrides during treatment and the resulting replenishment quantities. On average, actual replenishment exceeds recommendations in all groups, which retain full discretion to increase quantities. Workers under free override remove $9.2$ items across $2.7$ SKUs per machine, while those under constrained override remove about $5.0$ items across $1.2$ SKUs. Thus, constrained workers reduce fewer SKUs but make larger reductions per selected SKU ($4.45$ vs. $3.42$). They also increase quantities for slightly more SKUs and add more inventory per machine and per selected SKU.

\begin{table}[!htbp] 
    \linespread{0.5}
    \centering 
    \scriptsize
    \caption{Summary Statistics of Post-Treatment Replenishment and Override Behavior}
    \label{table:override_analysis} 
    \begin{tabular}{@{\extracolsep{2pt}}lccc} 
        \toprule
         & \textit{No {Downward}} & \textit{Constrained {Downward}} & \textit{Free {Downward}} \\ 
         & \textit{Override} & \textit{Override} & \textit{Override} \\
        
        \midrule
        
        \textbf{Algorithm Recommended}\\ [0.5ex] 
        
        \textit{Total Replenishment Quantity per Machine} 
         & $67.5030\ [23.8934]$ & $67.3102\ [25.8881]$ & $67.7611\ [24.6702]$ \\
         
        \textit{\# of Replenished SKUs per Machine} 
         & $15.0123\ [3.9810]$ & $14.7256\ [4.6331]$ & $14.6471\ [4.0278]$ \\ 
        
        \textit{Replenishment Quantity per SKU} 
         & $4.0221\ [0.9055]$ & $4.0059\ [1.1946]$ & $4.1116\ [1.2644]$ \\ 
        
        \textbf{Actual} \\  [0.5ex] 
        
        \textit{Total Replenishment Quantity per Machine} 
         & $74.7499\ [25.6278]$ & $74.1293\ [24.4824]$ & $68.9102\ [23.9915]$ \\ 
        
        \textit{\# of Replenished SKUs per Machine} 
         & $16.2870\ [3.7708]$ & $16.2705\ [3.6787]$ & $15.5820\ [3.6972]$ \\

        \textit{Replenishment Quantity per SKU} 
         & $4.4988\ [1.1620]$ & $4.5013\ [1.4862]$ & $4.3165\ [1.2964]$ \\ 
        
        \midrule
        
        \textbf{Downward Override} \\[0.5ex]
        
        \textit{
        \textbf{\# of SKUs Overridden per Machine}} 
         & \textbf{$0.6627$} $[1.9577]$ & \textbf{$1.1795$} $[2.0567]$ & \textbf{$2.7275$} $[2.9104]$ \\ 
        
        \textit{Quantity Reduced per Overridden SKU} 
         & $3.1414\ [2.4278]$ & $4.4487\ [2.6788]$ & $3.4221\ [1.8185]$ \\ 
        
        \textit{Share of SKUs Decreased to Zero Quantity per Machine} 
         & $0.0119\ [0.0581]$ & $0.0227\ [0.0613]$ & $0.0358\ [0.0569]$ \\ 
        
        \textit{Total Replenishment Quantity Reduced per Machine} 
         & $2.4203\ [7.5783]$ & $4.9991\ [7.9291]$ & $9.1937\ [10.4052]$ \\ 
        
        \textbf{Upward Override} \\  [0.5ex] 
        
        \textit{\# of SKUs Overridden per Machine } 
        & $3.1416\ [3.8892]$ & $3.6084\ [4.4969]$ & $3.4380\ [3.7505]$ \\ 
        
        \textit{Quantity Increased per Overridden SKU} 
         & $2.9616\ [1.9653]$ & $3.2681\ [3.2192]$ & $2.9593\ [1.8809]$ \\ 
        
        \textit{Total Replenishment Quantity Increased per Machine} 
         & $9.6672\ [14.5975]$ & $11.8182\ [17.6974]$ & $10.3429\ [13.4644]$ \\ 
        
        \bottomrule \\[-0.5ex]
        \multicolumn{4}{l}{\textit{Note.} Quantities are in items; brackets report standard deviations. Statistics on overridden SKUs per machine include machines with no overrides.} \\
        \multicolumn{4}{l}{Quantity changes per overridden SKU are averaged over machine replenishments with at least one override in the corresponding direction.}
    \end{tabular} 
    \vspace{-15pt}
\end{table}

The DID estimates in \Cref{table:effect_on_behavior} of \Cref{section: results} complement these descriptive comparisons. Constrained override also increases the number of SKUs for which workers raise replenishment quantities, with no statistically significant change in the total quantity added through upward overrides. These findings describe workers' choices within each visit. We next examine how algorithmic recommendations and workers' decisions change across visits.

\subsection{Algorithmic Recommendations and Human Learning} \label{subsection:algorithm_suggestions_learning}
\textbf{Algorithmic Recommendations.} We first examine whether changes in algorithmic recommendations help explain the main policy effects. The forecasting procedure, target service level, and replenishment rule were the same across groups throughout the experiment. Even with these rules unchanged, recommended quantities can respond to earlier decisions: when a worker leaves less inventory, the system may recommend more at the next visit simply to refill the machine. We therefore examine both recommended quantity and the inventory that would be available after following the recommendation. We call the latter the \textit{implied target}, calculated as beginning inventory plus recommended quantity. This distinction helps separate refilling the machine from raising its inventory target. Recommendations are subject to machine capacity and warehouse availability. We estimate weekly policy effects on both measures using \eqref{eq:did_machine_dynamic}.

\Cref{table:algorithm_weekly_did} shows no statistically significant change in the implied target under either policy in any treatment week. Under free override, recommended quantity increases by approximately $3.3$--$4.5$ items from Week~1 onward, following workers' larger reductions. Under constrained override, the recommendation estimates are not statistically significant except in Week~3 ($p<0.10$). These patterns are consistent with the system replenishing inventory removed by earlier overrides, with no detectable change in the implied inventory target.

\begin{table}[!htbp]
    \linespread{0.5}
    \centering
    \scriptsize
    \caption{Weekly Effects on Algorithmic Recommendations and Worker Overrides}
    \label{table:algorithm_weekly_did}
    \begin{tabular}{@{\extracolsep{4pt}}lcccc}
        \toprule
        & \makecell{Recommended\\Quantity}
        & \makecell{Implied\\Target}
        & \makecell{Downward Override\\Quantity}
        & \makecell{Upward Override\\Quantity}\\
        & (\textit{Unit: Item})
        & (\textit{Unit: Item})
        & (\textit{Unit: Item})
        & (\textit{Unit: Item})\\
        \midrule

        \multicolumn{5}{l}{\textbf{Constrained override}}\\
        Week~0 & $0.0199\ (1.0510)$ & $-1.1110\ (1.8520)$ & $1.2970^{***}\ (0.3643)$ & $0.3085\ (0.5852)$\\

        Week~1 & $1.1870\ (1.0770)$ & $1.8830\ (1.8190)$ & $1.7830^{***}\ (0.4270)$ & $0.8386\ (0.5685)$\\

        Week~2 & $0.7990\ (1.0640)$ & $-0.6091\ (1.8330)$ & $1.7290^{***}\ (0.4570)$ & $0.3102\ (0.5236)$\\

        Week~3 & $1.9990^{*}\ (1.1640)$ & $1.7290\ (1.9910)$ & $2.0120^{***}\ (0.4713)$ & $0.1450\ (0.6177)$\\

        Week~4 & $0.7336\ (1.3330)$ & $-0.6005\ (2.2810)$ & $2.5710^{***}\ (0.4789)$ & $0.6598\ (0.6313)$\\

        Week~5 & $0.0005\ (1.3460)$ & $-2.1740\ (2.4400)$ & $3.1600^{***}\ (0.6476)$ & $0.7266\ (0.6428)$\\

        Week~6 & $0.7487\ (1.5160)$ & $-1.8180\ (2.8030)$ & $2.9310^{***}\ (0.6758)$ & $0.4027\ (0.7805)$\\

        \addlinespace[0.5ex]

        \multicolumn{5}{l}{\textbf{Free override}}\\
        Week~0 & $0.7786\ (1.5800)$ & $-3.3310\ (2.4770)$ & $4.3670^{***}\ (0.7183)$ & $-0.7129\ (0.7843)$\\

        Week~1 & $4.3360^{**}\ (1.7830)$ & $2.1290\ (2.6670)$ & $5.5250^{***}\ (0.7978)$ & $-0.8700\ (0.8585)$\\

        Week~2 & $3.6650^{**}\ (1.8500)$ & $0.1358\ (2.6680)$ & $5.6120^{***}\ (0.8844)$ & $-0.6058\ (0.9073)$\\

        Week~3 & $3.3010^{*}\ (1.7480)$ & $-0.2035\ (2.5370)$ & $6.0120^{***}\ (0.8486)$ & $-0.3197\ (0.9067)$\\

        Week~4 & $3.5810^{*}\ (2.0950)$ & $-2.0170\ (2.9630)$ & $8.3400^{***}\ (0.9318)$ & $-0.5714\ (0.9678)$\\

        Week~5 & $3.8800^{*}\ (2.0330)$ & $-2.3510\ (2.9750)$ & $9.2980^{***}\ (1.0200)$ & $-0.6456\ (0.9941)$\\

        Week~6 & $4.4860^{**}\ (2.2520)$ & $-1.9040\ (3.5190)$ & $8.8510^{***}\ (1.0210)$ & $-1.4540\ (1.0640)$\\

        \midrule
        Worker, (City--Date) FE
        & Yes & Yes & Yes & Yes\\

        $R^2$
        & $0.1987$
        & $0.1453$
        & $0.3691$
        & $0.3012$\\

        Observations
        & $639{,}112$
        & $639{,}112$
        & $639{,}112$
        & $639{,}112$\\

        \bottomrule\\[-0.8ex]

        \multicolumn{5}{p{0.85\textwidth}}{\textit{Note.} Rows report policy-by-week interactions; Week~0 is the first treatment week. Robust standard errors clustered at the worker level are reported in parentheses. $^{*}p<0.1$; $^{**}p<0.05$; $^{***}p<0.01$.}\\
    \end{tabular}
\end{table}

To check whether changes in the implied target account for the main findings, we add it to the performance regressions. Because the target can respond to earlier decisions, we use this specification as a robustness check and retain the randomized policy comparison as our main estimate. In \Cref{table:performance_implied_target}, the constrained-override inventory coefficient is $-0.0126$, compared with $-0.0128$ in \Cref{table: main results}. Its sales estimates remain statistically insignificant. Under free override, the sales coefficient changes from $-0.0119$ to $-0.0116$, and the inventory coefficient from $-0.0195$ to $-0.0192$; both remain significantly negative. Accounting for the measured target therefore leaves the main sales--inventory findings essentially unchanged.

\begin{table}[!htbp]
    \linespread{0.5}
    \centering
    \scriptsize
    \caption{Performance Effects of Free and Constrained Override Policies Including Implied Target}
    \label{table:performance_implied_target}
    \begin{tabular}{@{\extracolsep{6pt}}lccc}
        \toprule
        & $\log(\textit{Sales}+1)$
        & Sales Probability
        & $\log(\textit{Inventory}+1)$\\[0.5ex]

        \midrule

        $ConstrainedOverride_i \cdot Post_t$ & $0.0012\ (0.0039)$ & $-0.0019\ (0.0029)$ & $-0.0126^{**}\ (0.0055)$\\

        $FreeOverride_i \cdot Post_t$ & $-0.0116^{**}\ (0.0048)$ & $-0.0136^{***}\ (0.0038)$ & $-0.0192^{**}\ (0.0079)$\\

        \textit{Implied Target}
        & $0.0011^{***}$
        & $0.0006^{***}$
        & $0.0010^{***}$\\

        & $(2.20\times10^{-5})$
        & $(1.18\times10^{-5})$
        & $(2.70\times10^{-5})$\\

        \midrule

        Worker fixed effects
        & Yes & Yes & Yes\\

        (City--Date) fixed effects
        & Yes & Yes & Yes\\

        R$^2$
        & $0.3106$
        & $0.2942$
        & $0.2501$\\

        Observations
        & $639{,}112$
        & $639{,}112$
        & $639{,}112$\\

        \bottomrule\\[-0.8ex]

        \multicolumn{4}{l}{\textit{Note.} Robust standard errors clustered at the
        worker level are reported in parentheses.}\\
        \multicolumn{4}{l}{$^{*}p<0.1$; $^{**}p<0.05$; $^{***}p<0.01$.}
    \end{tabular}
\end{table}

\textbf{Human Learning.} Workers may learn from the outcomes of their earlier decisions. Zero inventory after reducing a recommendation may encourage greater caution next time, whereas substantial remaining inventory may suggest room for further reductions. We examine these responses by linking consecutive replenishments of the same SKU at the same machine by the same worker, with both visits occurring during treatment. We ask whether workers respond differently when the previous outcome follows a downward override they chose.

For a current decision $(i,j,k,t)$, let $t^{-}$ denote the previous visit. The indicator $\mathit{DownOverride}_{ijk,t^{-}}$ records whether the worker reduced the previous recommendation. The indicators $\mathit{ZeroInventory}_{ijk,t^{-}}$ and $\mathit{HighInventory}_{ijk,t^{-}}$ record whether inventory immediately before the current visit was zero or high. We compare both outcomes with positive but not high inventory; the cutoff for high inventory is defined in the table notes. Let $g(i)\in\{\mathrm{Free},\mathrm{Constrained}\}$ denote worker $i$'s policy assignment. We estimate:
\begin{equation*}
    \begin{aligned}
        Y_{ijkt} ={}& \alpha +\beta_{1,g(i)}\mathit{DownOverride}_{ijk,t^{-}} +\beta_{2,g(i)}\mathit{ZeroInventory}_{ijk,t^{-}} +\beta_{3,g(i)}\mathit{HighInventory}_{ijk,t^{-}}\\
        &+\beta_{4,g(i)} \left( \mathit{DownOverride}_{ijk,t^{-}} \times \mathit{ZeroInventory}_{ijk,t^{-}} \right)\\ &+\beta_{5,g(i)} \left( \mathit{DownOverride}_{ijk,t^{-}} \times \mathit{HighInventory}_{ijk,t^{-}} \right)\\
        &+X_{ijkt}'\theta +\gamma_{ij} +\lambda_k +\delta_{c(j)t} +\epsilon_{ijkt}.
    \end{aligned}
\end{equation*}

The outcomes are whether the worker makes a downward override, the quantity reduced, and whether the current override is safe. Quantity reduced is zero when no downward override is made. Safety is measured only for current downward overrides with an observed next replenishment by the same worker for the same machine--SKU. The controls $X_{ijkt}$ include current inventory, recent sales, and recommended quantity; $\gamma_{ij}$, $\lambda_k$, and $\delta_{c(j)t}$ are worker--machine, SKU, and city--date fixed effects. Previous decisions and outcomes are not randomized, so the estimates describe adjustments consistent with learning.

\Cref{table:worker_learning_previous_outcome} reports $\beta_{4,g(i)}$ and $\beta_{5,g(i)}$. We first compare workers' responses to zero or high inventory with their responses to positive but not high inventory. The reported coefficients show how these responses differ when workers had reduced the previous recommendation, compared with when they had not.
\begin{table}[!htbp]
    \vspace{-5pt}
    \centering
    \scriptsize
    \setlength{\tabcolsep}{3.5pt}
    \renewcommand{\arraystretch}{0.95}
    \caption{Worker Responses to Previous Replenishment Outcomes}
    \label{table:worker_learning_previous_outcome}

    \begin{tabular}{@{}lcccccc@{}}
        \toprule

        & \multicolumn{2}{c}{Downward Override}
        & \multicolumn{2}{c}{\makecell{Downward Override Quantity\\(\textit{Unit: Item})}}
        & \multicolumn{2}{c}{Safe Override}\\

        \cmidrule(lr){2-3}
        \cmidrule(lr){4-5}
        \cmidrule(lr){6-7}

        & Free
        & Constrained
        & Free
        & Constrained
        & Free
        & Constrained\\

        \midrule

        \makecell[l]{Previous Downward Override\\$\times$ Zero Ending Inventory}
        & \makecell{$-0.0263$\\$(0.0209)$}
        & \makecell{$-0.0508^{*}$\\$(0.0304)$}
        & \makecell{$0.0906$\\$(0.1585)$}
        & \makecell{$-0.2664^{**}$\\$(0.1234)$}
        & \makecell{$0.1079$\\$(0.0680)$}
        & \makecell{$0.1259^{***}$\\$(0.0410)$}\\

        \makecell[l]{Previous Downward Override\\$\times$ High Ending Inventory}
        & \makecell{$0.0090$\\$(0.0066)$}
        & \makecell{$0.0246^{**}$\\$(0.0095)$}
        & \makecell{$0.0148$\\$(0.0570)$}
        & \makecell{$0.2140^{***}$\\$(0.0733)$}
        & \makecell{$0.0007$\\$(0.0028)$}
        & \makecell{$0.0042$\\$(0.0035)$}\\

        \midrule

        (Worker--Machine), SKU, (City--Date) FE & \multicolumn{2}{c}{Yes} & \multicolumn{2}{c}{Yes} & \multicolumn{2}{c}{Yes} \\

        $R^2$
        & \multicolumn{2}{c}{$0.6233$}
        & \multicolumn{2}{c}{$0.6774$}
        & \multicolumn{2}{c}{$0.6404$}\\

        Observations
        & \multicolumn{2}{c}{$1{,}153{,}653$}
        & \multicolumn{2}{c}{$1{,}153{,}653$}
        & \multicolumn{2}{c}{$124{,}326$}\\

        \bottomrule\\[-0.5ex]

        \multicolumn{7}{p{0.95\textwidth}}{
        \textit{Note.}
            Entries show how the response to zero or high ending inventory differs after a previous downward override versus no previous downward override. Both responses use positive but not high inventory as the reference outcome. High ending inventory is defined using the SKU-specific 75th percentile of positive pre-treatment ending inventory. For SKUs with fewer than 20 such observations, we use the 75th percentile calculated from the combined pre-treatment observations across all SKUs. The safe-override specifications are restricted to current downward overrides with an observed next-cycle outcome. Robust standard errors clustered at the worker level are reported in parentheses. $^{*}p<0.1$; $^{**}p<0.05$; $^{***}p<0.01$.
        }\\

    \end{tabular}
    \vspace{-10pt}
\end{table}

Under constrained override, workers are more cautious after zero inventory when it follows their own earlier reduction: the additional changes are $-5.08$ percentage points in subsequent override probability ($p<0.10$) and $-0.27$ items in quantity reduced ($p<0.05$). After high inventory, the corresponding additional changes are $2.46$ percentage points ($p<0.05$) and $0.21$ items ($p<0.01$). Among subsequent downward overrides, the additional change in safe-override probability after zero inventory is also positive ($p<0.01$). The free-policy coefficients and the high-inventory safety coefficient are not statistically significant. These patterns are consistent with workers under constrained override using the outcomes of their own reductions to adjust later decisions.

These adjustments raise the question of whether policy benefits emerge gradually as workers gain experience. The weekly estimates in \Cref{table:main_results_by_week} include any learning between visits within each week. Constrained override's inventory benefit is already present in the first treatment week (Week~0): the coefficient on $\log(\textit{Inventory}+1)$ is $-0.0155$ ($p<0.01$), without statistically significant sales or sales-probability changes. Its inventory estimates remain negative throughout the experiment, with no steady increase in the size of the benefit, and its sales-quantity estimates are insignificant in every week. Under free override, the sales-probability loss also appears in Week~0 and remains significant through Week~5. This early and persistent policy contrast is consistent with workers prioritizing overrides from the beginning while continuing to learn.

\section{Additional Analyses for Section~\ref{section:selective filtering}} \label{appendix:selective filtering test}
Figure~\ref{fig: late} groups instances by policy assignment $D_i\in\{0,1\}$ and observed override status $\text{DownOverride}_{ijkt}=\text{DownOverride}_{ijkt}(1)\cdot D_i+\text{DownOverride}_{ijkt}(0)\cdot(1-D_i)$. Overrides under no override (top-left) typically reflect operational exceptions such as warehouse stockouts. Under monotonicity, these are always-overridden cases: they would also be overridden under constrained override. Overridden cases under constrained override (top-right) include both always-overridden and complier instances. By randomization, the no-override group also contains compliers whose recommendations are followed but would be reduced under constrained override (bottom-left). Under the IV assumptions, LATE captures the causal effect for these policy-induced override decisions, which are the decisions the firm can influence.

\begin{figure}[htbp]
    \FIGURE
    {\includegraphics[width=0.55\textwidth]{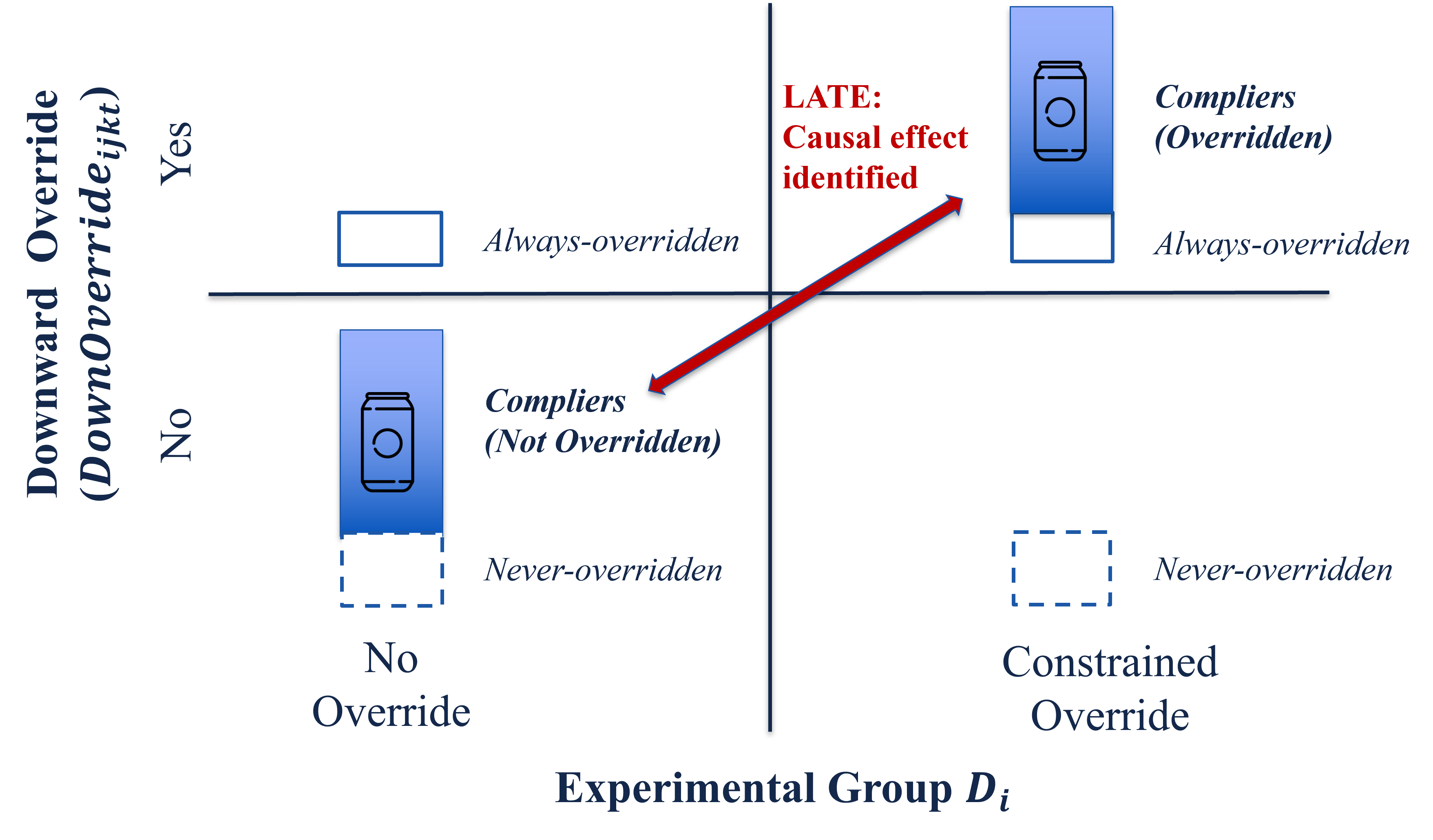}}
    {Illustration of LATE Identification: Constrained vs. No Override Policy\label{fig: late}}
    {}
    \vspace{-10pt}
\end{figure}

\subsection{LATE Assumptions}\label{appendix:exclusion}
The SKU-level LATE analysis uses randomized policy assignment as an instrument for downward override. Its interpretation rests on four assumptions. \textit{Random assignment} makes the assigned policy independent of workers' prior characteristics and potential outcomes. \Cref{table: randomization check} reports the pre-treatment balance checks. \textit{Relevance} requires assignment to change the probability of a downward override. In \Cref{table:late down}, constrained and free override increase this probability by $2.87$ and $12.92$ percentage points, respectively, relative to no override (both $p<0.01$). The corresponding first-stage $F$-statistics are $17.68$ and $85.06$. \textit{Monotonicity} requires a downward override made under no override to also occur when workers receive greater discretion: $\text{DownOverride}_{ijkt}(1)\geq\text{DownOverride}_{ijkt}(0)$. Under no override, reductions are permitted only for operational exceptions, such as warehouse stockouts. The active policies retain these exceptions, so a reduction required for such a reason would generally remain necessary when discretion is expanded. 

The \textit{exclusion restriction} requires the assigned policy to affect a SKU's outcome only by changing whether its own recommendation is reduced. The treatment records whether a recommendation is reduced, not how much it is reduced. The interpretation therefore also assumes that the policy does not affect outcomes by changing the size of reductions that would occur under both policies. We exclude upward overrides from the sample. If the policy changes which decisions are excluded, the remaining groups may no longer be comparable.

Reducing one SKU may lead customers to buy another SKU instead, or affect sales of products bought together. A SKU's sales could then change even when its own recommendation is followed \citep{chen2024new}. We use three checks to assess this concern. The first two combine outcomes across SKUs within a machine and across machines managed by the same worker within the same city. The third compares outcomes where workers follow the recommendation. These checks provide evidence about the concern, although they cannot establish the assumption.

First, reducing one SKU may lead customers to buy another SKU in the same machine. Looking only at the first SKU would miss any sales gained by the second. We therefore average the SKU-level sales and inventory measures across the included SKU observations within each machine replenishment, using the main LATE sample. Machine-level sales probability is the share of these SKU observations with at least one unit sold during the replenishment window. This lets us check whether the results still hold when changes in other SKUs are included. We use the number of downward-overridden SKUs within each machine replenishment as the treatment measure. The IV coefficients report differences in average machine outcomes between policy groups, expressed per additional downward-overridden SKU. We use the machine- and worker-level estimates as suggestive evidence on outcome patterns.

\begin{table}[!htbp]
    \vspace{-10pt}
    \linespread{0.5}
    \centering
    \scriptsize
    \setlength{\tabcolsep}{5pt}
    \renewcommand{\arraystretch}{0.95}
    \caption{Machine-Level Aggregate IV Estimates}
    \label{table: shelf late down}
    \begin{tabular}{@{}lcccc@{}}
        \toprule[1pt]
        & \multicolumn{1}{c}{\textit{First Stage}}
        & \multicolumn{3}{c}{\textit{Second Stage (IV)}}\\
        \cmidrule(lr){2-2}\cmidrule(lr){3-5}
        & \makecell{\textit{Downward-}\\\textit{Override SKUs}}
        & \makecell{\textit{Sales}\\\textit{(Unit: Item)}}
        & \textit{Sales Probability}
        & \makecell{\textit{Inventory}\\\textit{(Unit: Item)}}\\
        \midrule

        \multicolumn{5}{@{}l}{%
            \textbf{\textit{Panel A. Constrained vs.\ No Override}}
        }\\[0.25ex]

        $ConstrainedOverride_i \times Post_t$ & $0.3982^{***}\ (0.1234)$ &  &  & \\[0.25ex]

        $\widehat{DownOverride}_{ijt}$ &  & $0.0570\ (0.0399)$ & $-0.0028\ (0.0069)$ & $-0.2513^{***}\ (0.0937)$\\[0.25ex]

        (Worker-Machine), (City-Date) FE   & Yes & Yes & Yes & Yes \\

        R$^2$
        & $0.6085$
        & $0.6161$
        & $0.6860$
        & $0.5877$\\

        First-stage F-statistic
        & $10.4147^{***}$
        &
        &
        & \\

        Observations
        & \multicolumn{4}{c}{$498{,}333$}\\

        \midrule

        \multicolumn{5}{@{}l}{%
            \textbf{\textit{Panel B. Free vs.\ No Override}}
        }\\[0.25ex]

        $FreeOverride_i \times Post_t$ & $1.8041^{***}\ (0.2132)$ &  &  & \\[0.25ex]

        $\widehat{DownOverride}_{ijt}$ &  & $-0.0011\ (0.0080)$ & $-0.0071^{***}\ (0.0018)$ & $-0.0938^{***}\ (0.0222)$\\[0.25ex]
        
        (Worker-Machine), (City-Date) FE   & Yes & Yes & Yes & Yes \\
        
        R$^2$
        & $0.6285$
        & $0.6443$
        & $0.6835$
        & $0.6189$\\

        First-stage F-statistic
        & $71.6307^{***}$
        &
        &
        & \\

        Observations
        & \multicolumn{4}{c}{$372{,}551$}\\

        \bottomrule
        
    \end{tabular}
    \par\smallskip
    \begin{minipage}{0.75\linewidth}
        \scriptsize
        \linespread{0.5}\selectfont
        \textit{Note.} Robust standard errors clustered at the worker
        level are in parentheses.\\
        $^{*}p<0.10$; $^{**}p<0.05$; $^{***}p<0.01$.
    \end{minipage}
    \vspace{-10pt}
\end{table}

In \Cref{table: shelf late down}, constrained override has a negative inventory estimate ($p<0.01$) and no statistically significant change in sales or sales probability. Free override has negative inventory and sales-probability estimates (both $p<0.01$), with no statistically significant change in sales quantity. The aggregate results provide complementary evidence relevant to the exclusion concern: the inventory and sales-probability patterns persist when outcomes across the included SKUs are considered together.

Second, workers may change how they divide stock or effort across machines. A gain at one machine could therefore come with a loss at another. To check for this possibility, we average each machine-level outcome across the included machine observations within each worker--city--date group, giving each machine observation equal weight. The treatment measure is the total number of downward-overridden SKU instances across these machines. In \Cref{table: worker late down}, constrained override has a negative inventory estimate ($p<0.05$), with no statistically significant change in sales or sales probability. Free override has negative sales-probability and inventory estimates ($p<0.01$), with an insignificant sales estimate. The constrained policy's inventory benefit and the free policy's lower sales probability therefore remain when we consider the worker's included machines together.

\begin{table}[!htbp]
    \vspace{-5pt}
    \linespread{0.5}
    \centering
    \scriptsize
    \setlength{\tabcolsep}{5pt}
    \renewcommand{\arraystretch}{0.95}
    \caption{Worker-Level Aggregate IV Estimates}
    \label{table: worker late down}
    \begin{tabular}{@{}lcccc@{}}
        \toprule[1pt]
        & \multicolumn{1}{c}{\textit{First Stage}}
        & \multicolumn{3}{c}{\textit{Second Stage (IV)}}\\
        \cmidrule(lr){2-2}\cmidrule(lr){3-5}
        & \makecell{\textit{Overall Downward-}\\\textit{Override SKUs}}
        & \makecell{\textit{Sales}\\\textit{(Unit: Item)}}
        & \textit{Sales Probability}
        & \makecell{\textit{Inventory}\\\textit{(Unit: Item)}}\\
        \midrule

        \multicolumn{5}{@{}l}{%
            \textbf{\textit{Panel A. Constrained vs.\ No Override}}
        }\\[0.25ex]

        $ConstrainedOverride_i \times Post_t$ & $5.9452^{***}\ (1.7087)$ &  &  & \\[0.25ex]

        $\widehat{DownOverride}_{ict}$ &  & $0.0028\ (0.0023)$ & $-0.0003\ (0.0004)$ & $-0.0118^{**}\ (0.0056)$\\[0.25ex]

        Worker, (City-Date) FE
        & Yes
        & Yes
        & Yes
        & Yes\\

        R$^2$
        & $0.5550$
        & $0.8594$
        & $0.9704$
        & $0.9215$\\

        First-stage F-statistic
        & $12.1059^{***}$
        &
        &
        & \\

        Observations
        & \multicolumn{4}{c}{$36{,}393$}\\

        \midrule

        \multicolumn{5}{@{}l}{%
            \textbf{\textit{Panel B. Free vs.\ No Override}}
        }\\[0.25ex]

        $FreeOverride_i \times Post_t$ & $26.9569^{***}\ (3.1828)$ &  &  & \\[0.25ex]

        $\widehat{DownOverride}_{ict}$ &  & $0.0002\ (0.0006)$ & $-0.0004^{***}\ (0.0001)$ & $-0.0044^{***}\ (0.0015)$\\[0.25ex]

        Worker, (City-Date) FE
        & Yes
        & Yes
        & Yes
        & Yes\\

        R$^2$
        & $0.6099$
        & $0.8511$
        & $0.9741$
        & $0.9205$\\

        First-stage F-statistic
        & $71.7340^{***}$
        &
        &
        & \\

        Observations
        & \multicolumn{4}{c}{$26{,}831$}\\

        \bottomrule

    \end{tabular}
    \par\smallskip
    \begin{minipage}{0.9\linewidth}
        \scriptsize
        \linespread{0.5}\selectfont
        \textit{Note.} Each observation is a worker--city--date group, indexed by worker $i$, city $c$, and date $t$. Robust standard errors clustered at the worker
        level are in parentheses. $^{*}p<0.10$; $^{**}p<0.05$; $^{***}p<0.01$.
    \end{minipage}
    \vspace{-10pt}
\end{table}

Third, if reducing one SKU affects another SKU, outcomes may differ across policies even for SKUs whose recommendations are followed. We compare these decisions across groups, controlling for inventory before replenishment (\Cref{table: sku nonoverride policy}). The policies may change which recommendations workers follow, and starting inventory may already reflect earlier overrides, so the comparison is descriptive. None of the sales, sales-probability, or inventory estimates is statistically significant under either policy. We therefore find no detectable policy differences for these SKUs. Together, the three checks reduce the concern that effects on other SKUs explain the main findings.

\begin{table}[!htbp]
    \vspace{-10pt}
    \linespread{0.5}
    \centering
    \scriptsize
    \caption{SKU-level Policy Coefficients with Non-Override Instances}
    \label{table: sku nonoverride policy}
    \begin{tabular}{lccc}
        \toprule
        & \textit{Sales} & \textit{Sales Probability} & \textit{Inventory} \\
        & (\textit{Unit: Item}) &  & (\textit{Unit: Item})\\
        \midrule 
        
        \textbf{\textit{Constrained vs No Override}} \\
        
        \midrule 
        
        $ConstrainedOverride_i \cdot Post_{t}$ & $0.0051\ (0.0077)$ & $0.0019\ (0.0019)$ & $-0.0043\ (0.0113)$\\
    
        \textit{Inventory Before Replenishment} & $0.0720^{***}\ (0.0011)$ & $0.0052^{***}\ (0.0001)$ & $0.8415^{***}\ (0.0020)$\\
        (Worker-Machine), SKU, (City-Date) FE   & Yes & Yes & Yes \\
        R$^2$                           & $0.3441$ & $0.3550$ & $0.9371$   \\
        Observations                    & $6{,}498{,}985$ & $6{,}498{,}985$ & $6{,}498{,}985$   \\
        
        \midrule
        
        \textbf{\textit{Free vs No Override}} \\
        
        \midrule 
        
        $FreeOverride_i \cdot Post_{t}$ & $-0.0026\ (0.0083)$ & $-0.0017\ (0.0023)$ & $0.0217\ (0.0149)$\\
        \textit{Inventory Before Replenishment} & $0.0729^{***}\ (0.0017)$ & $0.0051^{***}\ (0.0001)$ & $0.8440^{***}\ (0.0026)$\\
        (Worker-Machine), SKU, (City-Date) FE   & Yes & Yes & Yes \\
        R$^2$                           & $0.3480$ & $0.3590$ & $0.9383$   \\
        Observations                    & $4{,}719{,}210$ & $4{,}719{,}210$ & $4{,}719{,}210$   \\
        
        \bottomrule\\[-0.8ex]
        \multicolumn{4}{l}{\footnotesize \textit{Note.} Robust standard errors clustered at the worker level.} \\
        \multicolumn{4}{l}{\footnotesize $^{*}p<0.1$; $^{**}p<0.05$; $^{***}p<0.01$}
    \end{tabular}
    \vspace{-10pt}
\end{table}

\subsection{Additional Robustness Checks for the LATE Analysis} \label{appendix:late-robustness}
We first estimate how policy assignment changes SKU-level outcomes, without dividing by its effect on override probability. These are the reduced-form estimates reported in Table~\ref{table:itt down}. Constrained override reduces inventory ($p<0.01$), with no statistically significant sales change; free override reduces inventory and sales probability (both $p<0.01$). The LATE estimates have the same signs because each reduced-form estimate is divided by the corresponding increase in override probability. This shows how the LATE results relate to the policy effects in the same sample.

\begin{table}[!htbp]
    \vspace{-10pt}
    \linespread{0.5}
    \centering
    \scriptsize
    \caption{Reduced-Form Estimates with Downward Override}
    \label{table:itt down}
    \begin{tabular}{lccc}
        \toprule
        & \textit{Sales} & \textit{Sales Probability} & \textit{Inventory} \\
        & (\textit{Unit: Item}) &  & (\textit{Unit: Item})\\
        \midrule 
        
        \textbf{\textit{Constrained vs No Override}} \\
        
        \midrule 
        
        $Constrained\ Override_i \cdot Post_{t}$ & $0.0081\ (0.0089)$ & $-0.0013\ (0.0027)$ & $-0.0990^{***}\ (0.0262)$\\
        (Worker-Machine), SKU, (City-Date) FE   & Yes & Yes & Yes \\
        R$^2$                           & $0.2881$ & $0.3621$ & $0.3198$   \\
        Observations                    & $7,012,190$ & $7,012,190$ & $7,012,190$   \\
        
        \midrule
        
        \textbf{\textit{Free vs No Override}} \\
        
        \midrule 
        
        $Free\ Override_i \cdot Post_{t}$ & $-0.0154\ (0.0103)$ & $-0.0136^{***}\ (0.0033)$ & $-0.1625^{***}\ (0.0395)$\\
        (Worker-Machine), SKU, (City-Date) FE   & Yes & Yes & Yes \\
        R$^2$                           & $0.2806$ & $0.3545$ & $0.3183$   \\
        Observations                    & $5,189,189$ & $5,189,189$ & $5,189,189$   \\
        
        \bottomrule\\[-0.8ex]
        \multicolumn{4}{l}{\footnotesize Robust standard errors clustered at the worker level.} \\
        \multicolumn{4}{l}{\footnotesize $^{*}p<0.1$; $^{**}p<0.05$; $^{***}p<0.01$}
    \end{tabular}
    \vspace{-15pt}
\end{table}

We also examine decisions in which workers increase replenishment, since all groups retain the right to make upward overrides. Table~\ref{table:upward itt} compares outcomes across policies for these decisions. Because the policies may change which decisions involve an upward override, this comparison describes the selected decisions rather than a causal effect on all replenishments. None of the sales, sales-probability, or inventory estimates is statistically significant, providing no detectable evidence of a policy difference in this sample.

\begin{table}[!htbp]
    \vspace{-10pt}
    \linespread{0.5}
    \centering
    \scriptsize
    \caption{Policy Comparisons among Upward Overrides}
    \label{table:upward itt}
    \begin{tabular}{lccc}
        \toprule
        & \textit{Sales} & \textit{Sales Probability} & \textit{Inventory} \\
        & (\textit{Unit: Item}) &  & (\textit{Unit: Item})\\
        \midrule 
        
        \textbf{\textit{Constrained vs No Override}} \\
        
        \midrule 
        
        $ConstrainedOverride_i \cdot Post_{t}$ & $0.0380\ (0.0277)$ & $0.0002\ (0.0036)$ & $-0.0640\ (0.0549)$\\
        (Worker-Machine), SKU, (City-Date) FE   & Yes & Yes & Yes \\
        R$^2$                           & $0.3235$ & $0.1980$ & $0.3787$   \\
        Observations                    & $1,639,664$ & $1,639,664$ & $1,639,664$   \\
        
        \midrule
        
        \textbf{\textit{Free vs No Override}} \\
        
        \midrule 
        
        $FreeOverride_i \cdot Post_{t}$ & $-0.0290\ (0.0217)$ & $-0.0002\ (0.0044)$ & $-0.0294\ (0.0716)$\\
        (Worker-Machine), SKU, (City-Date) FE   & Yes & Yes & Yes \\
        R$^2$                           & $0.3077$ & $0.1994$ & $0.3648$   \\
        Observations                    & $1,187,107$ & $1,187,107$ & $1,187,107$   \\
        
        \bottomrule\\[-0.8ex]
        \multicolumn{4}{l}{\footnotesize Robust standard errors clustered at the worker level.} \\
        \multicolumn{4}{l}{\footnotesize $^{*}p<0.1$; $^{**}p<0.05$; $^{***}p<0.01$}
    \end{tabular}
    \vspace{-10pt}
\end{table}

\section{Simulation Setup for Section~\ref{section:betteroverride}} \label{appendix:better override}
To evaluate alternative caps $K$, we start with downward overrides observed under free override, where workers could reduce as many SKUs as they wished. We use constrained-group behavior to approximate which of these overrides workers would keep under a cap. The simulation proceeds in three steps.

\noindent\textbf{Step 1: Rank override opportunities.}
We pool downward overrides from the constrained and free groups, indexed by worker $i$, machine $j$, SKU $k$, and replenishment day $t$. Let $y_{ijkt}=1$ for constrained-group overrides and $0$ for free-group overrides. The characteristics $x_{ijkt}$ include current inventory, weekly sales, recommended and actual replenishment quantities, worker tenure, and machine and SKU characteristics. A logistic regression yields $\widehat p_{ijkt}=\widehat{\Pr}(y_{ijkt}=1\mid x_{ijkt})$. For each free-group override, this estimated probability measures how closely it resembles constrained-group overrides. We use it to rank observed overrides, holding the ranking fixed across caps to approximate which overrides workers would retain under each cap.

\noindent\textbf{Step 2: Apply each override cap.}
Within each free-group machine replenishment, we rank downward overrides from highest to lowest $\widehat p_{ijkt}$ and retain the top $K$, or all if fewer are observed. We restore the algorithm's recommendation for the remaining downward overrides and leave all other quantities unchanged. The simulated replenishment quantity is $q_{ijkt}(K) = q^A_{ijkt}$ for downward overrides not retained under the cap and $\widetilde q_{ijkt}$ otherwise. Thus, $K=0$ removes all downward overrides, providing the no-downward-override benchmark; $K=2$ applies the experimental cap to free-group opportunities; and a sufficiently large $K$ reproduces all free-group replenishment decisions.

\noindent\textbf{Step 3: Calculate sales, inventory, and profit.}
Starting from observed inventory immediately before replenishment, we add $q_{ijkt}(K)$ and simulate sales and inventory day by day until the next replenishment. For each SKU, we use the company's daily demand forecasts to construct a single demand path over this window, without drawing alternative demand realizations. We hold this forecast path and the observed replenishment window fixed across caps. On each day, simulated sales equal the minimum value of that day's predicted demand and available inventory. \textit{Sales} equals total simulated sales divided by the window length in days. The \textit{Sales Probability} indicator equals one if at least one unit is sold during the window and zero otherwise. \textit{Inventory} is the average quantity on hand over the entire window.  We average these SKU-level outcomes within each machine replenishment and calculate changes relative to $K=0$, denoted by $\Delta\textit{Sales}_{K}$ and $\Delta\textit{Inventory}_{K}$. Following \Cref{equation: replenishment value}, the simulated profit change is $ \Delta\textit{Profit}_{K}(r) = r\Delta\textit{Sales}_{K} - (1-r)\Delta\textit{Inventory}_{K}$, and $K^{*}(r)$ denotes the cap that yields the highest simulated profit under these forecast paths. We repeat the procedure by worker seniority and machine--SKU familiarity for the analyses in \Cref{fig:service-level-heterogeneity}. 

\begin{figure}[htbp]
    \vspace{-5pt}
    \edit{
    \FIGURE
    {
        \includegraphics[width=0.7\linewidth]{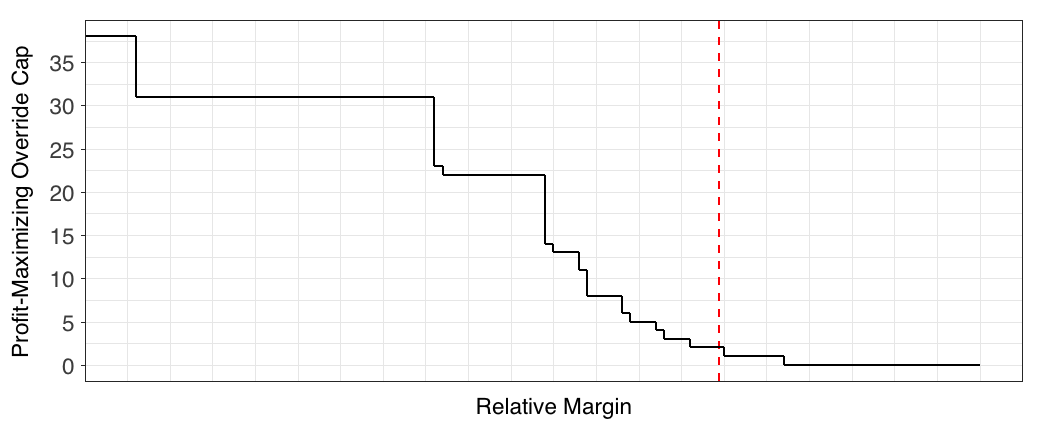}
    }
    {
        Simulated Profit-Maximizing Override Cap across Relative Margins\label{fig:service-level}
    }
    {
        The vertical dashed line denotes the company's relative margin.
    }
    }
    \vspace{-10pt}
\end{figure}

\Cref{fig:service-level} shows that the simulated profit-maximizing cap $K^*(r)$ decreases or stays unchanged as $r$ rises: lost sales become more costly relative to inventory savings. At the company's relative margin, the simulation selects $K=2$, matching the experimental cap. This simulated optimum depends on the company's demand forecasts and our approximation of which overrides workers would retain under alternative caps. The results provide guidance on caps to test rather than establish that $K=2$ is optimal; a lower cap such as $K=1$ remains an alternative for experimental evaluation.

\end{APPENDICES}

\end{document}